\documentclass[twocolumn,10pt]{IEEEtran}
\usepackage{cite}
\usepackage{graphicx}
\usepackage{amsmath}
\usepackage{amssymb}
\usepackage{bm}
\usepackage{amsthm}
\usepackage{cases}
\usepackage{bigstrut}
\usepackage{algorithmic}
\usepackage{multirow}
\usepackage{tabulary}
\usepackage{placeins}
\usepackage{subcaption}
\usepackage{bigstrut}
\usepackage{bbm}
\usepackage{mathtools}
\usepackage{color}

\newtheorem{lemma}{Lemma}
\newtheorem{theorem}{Theorem}
\newtheorem{definition}{Definition}

\newtheorem{proposition}{Proposition}

\theoremstyle{remark}

\begin{document}
	\title{Secure OFDM System Design and Capacity Analysis under Disguised Jamming}
		\author{\IEEEauthorblockN{Yuan Liang ~~~ Jian Ren ~~~ Tongtong Li}\\
			\IEEEauthorblockA{
				Department of Electrical \& Computer Engineering, Michigan State University \\
				Email: \{liangy11, renjian, tongli\}@egr.msu.edu
			}}
	
	\maketitle	
	
	\begin{abstract}
		In this paper, we propose a securely precoded OFDM (SP-OFDM) system for efficient and reliable transmission under disguised jamming, where the jammer intentionally misleads the receiver by mimicking the characteristics of the authorized signal, and causes complete communication failure. More specifically, we bring off a dynamic constellation by introducing secure shared randomness between the legitimate transmitter and receiver, and hence break the symmetricity between the authorized signal and the  disguised jamming. We analyze the channel capacities of both the traditional OFDM and SP-OFDM under hostile jamming using the arbitrarily varying channel (AVC) model. It is shown that the deterministic coding capacity of the traditional OFDM is zero under the worst disguised jamming. On the other hand, due to the secure randomness shared between the authorized transmitter and receiver, SP-OFDM can achieve a positive capacity under disguised jamming since the AVC channel corresponding to SP-OFDM is not symmetrizable. A remarkable feature of the proposed SP-OFDM scheme is that while achieving strong jamming resistance, it has roughly the same high spectral efficiency as the traditional OFDM system. The robustness of the proposed SP-OFDM scheme under disguised jamming is demonstrated through both theoretic and numerical analyses.
	\end{abstract}
	
	
	\section{Introduction}	
	
	In wireless systems, one of the most commonly used techniques for limiting the effectiveness of an opponent's communication is referred to as jamming, in which the authorized user's signal is deliberately interfered by the adversary. Along with the wide spread of various wireless devices, especially with the advent of user configurable intelligent devices, jamming attack is no longer limited to battlefield or military related events, but has become an urgent and serious threat to civilian communications as well.
	
	In literature \cite{Basar1983TIT, Medard1997, Kashyap2004TIT, Song2016TC }, jamming has widely been modeled as Gaussian noise. Based on the noise jamming model and the Shannon capacity formula, $C = B \log(1 + SNR),$  an intuitive impression is that jamming is really harmful only when the jamming power is much higher than the signal power. However, this is only partially true. More recently, it has been found that  disguised jamming \cite{Lapidoth1998TIT, Song2016TIFS, Zhang2013TWC2, Song2014Globalcom}, where the jamming is highly correlated with the signal, and has a power level close or equal to the signal power,  can be much more destructive than the noise jamming;  it can reduce the system capacity to zero even when the jamming power equals the signal power. Consider the following example:
	\begin{equation}
		R = S + J + N \nonumber
	\end{equation}
	where $S$ is the authorized signal, $J$ the jamming interference, $N$ the noise independent of $J$ and $S$, and $R$ the received signal. If the jammer is capable of eavesdropping on the symbol constellation and the codebook of the transmitter, it can simply replicate one of the sequences in the codebook of the legitimate transmitter, the receiver, then, would not be able to distinguish between the authorized sequence and the jamming sequence, resulting in a complete communication failure \cite[ch 7.3]{Gamal2012}.

	Orthogonal frequency division multiplexing (OFDM), due to its high spectral efficiency and robustness under fading channels, has been widely used in modern high speed multimedia communication systems \cite{Hwang2009TVT}, such as LTE and WiMax. However, unlike the spread spectrum techniques \cite{CDMA}, OFDM mainly relies on channel coding for communication reliability under hostile jamming, and has very limited built-in resilience against jamming attacks~\cite{Jun2007WTS, Amuru2015TIFS, Mailaender2013, Clancy2011ICC, Cuccaro2017, Pan2012, Lightfoot2009JASP}. For example, in \cite{Jun2007WTS}, the bit error rate (BER) performance of the traditional OFDM was explored under full-band and partial band Gaussian jamming, as well as multitone jamming. It was shown that  OFDM is quite fragile under jamming, as BER can go above $10^{-1}$ when the jamming power is the same as the signal power. In \cite{Clancy2011ICC, Cuccaro2017, Pan2012}, the jamming attacks aiming at the pilots in OFDM systems were studied. It was shown that when the system standard is public and no encryption is applied to the transmitted symbol sequence, pilot attacks can completely nullify the channel estimation and synchronization of OFDM, and hence result in complete communication failure. Most existing work \cite{Jun2007WTS, Pan2012, Amuru2015TIFS} has been focused on the jamming attacks which damage OFDM by minimizing the signal-to-interference power ratio (SIR). In this paper, we identify the threat to OFDM from the disguised jamming: when the jamming interference is also OFDM modulated, the receiver can easily be deceived into synchronizing with the jamming interference instead of the legitimate signal, hence paralyzing the legitimate transmission.
	
	In \cite{Mailaender2013}, the anti-jamming performance of Frequency Hopped (FH) OFDM system was explored. Like the traditional FH system, this approach achieves jamming resistance through large frequency diversity and sacrifices the spectral efficiency of OFDM.  In \cite{Lightfoot2009JASP}, a collision-free frequency hopping (CFFH) scheme was proposed, where the basic idea was to randomize the jamming interference through frequency domain interleaving based on secure, collision-free frequency hopping. The most significant feature of CFFH based OFDM is that it is very effective under partial band jamming, and at the same time, has the same spectral efficiency as the original OFDM.  However, CFFH based OFDM is still fragile under \emph{disguised jamming} \cite{Song2016TIFS, Zhang2013TWC2, Song2014Globalcom, Ericson1986TIT}.

	To combat disguised jamming in OFDM systems, a precoding scheme was proposed in \cite{Song2014Globalcom}, where extra redundancy is introduced to achieve jamming resistance. However, lack of plasticity in the precoding scheme results in inadequate reliability under cognitive disguised jamming. As OFDM being identified as a major modulation technique for the 5G systems, there is an ever increasing need on the development of  secure and efficient OFDM systems that are reliable under hostile jamming, especially the destructive disguised jamming.

	If we examine disguised jamming carefully, we can see that the main issue there is the symmetricity between the authorized signal and the jamming interference. Intuitively, to design the corresponding anti-jamming system, the main task is to break the symmetricity between the authorized signal and the jamming interference, or make it impossible for the jammer to achieve this symmetricity. For this purpose, encryption or channel coding at the bit level will not really help, since the symmetricity appears at the symbol level. That is, instead of using a fixed symbol constellation, we have to introduce secure randomness to the constellation, and utilize a dynamic constellation scheme, such that the jammer can no longer mimic the authorized user's signal. At the same time, the authorized user does not have to sacrifice too much on the performance, efficiency and system complexity.
    
    Motivated by the observations above and our previous research on anti-jamming system design~\cite{Song2014Globalcom, Song2016TIFS, Zhang2013TWC1, Zhang2013TWC2, Lightfoot2009JASP}, in this paper, we propose a securely precoded OFDM (SP-OFDM) system for efficient and reliable transmission under disguised jamming. By integrating advanced cryptographic techniques into OFDM transceiver design, we design a dynamic constellation  by  introducing shared randomness between the legitimate transmitter and receiver, which breaks the symmetricity between the authorized signal and the jamming interference, and hence ensures reliable performance under disguised jamming. A remarkable feature of the proposed SP-OFDM scheme is that it achieves strong jamming resistance, but has the same high spectral efficiency as the traditional OFDM system. Moreover, the change to the physical layer transceivers is minimal, feasible and affordable. The robustness of the proposed SP-OFDM scheme under disguised jamming is demonstrated through both theoretic and numerical analyses.
	
	More specifically, the main contributions of this paper can be summarized as follows:
	
	\begin{itemize}
		
		\item We design a highly secure and efficient OFDM system under disguised jamming, named securely precoded OFDM (SP-OFDM),   by exploiting secure symbol-level precoding  basing on phase randomization. The basic idea is to randomize the phase of transmitted symbols using the secure PN sequences generated from the Advanced Encryption Standard (AES) algorithm. The security is guaranteed by the secret key shared only between  the legitimate transmitter and receiver. While SP-OFDM achieves strong jamming resistance, it does not introduce too much extra coding redundancy into the system and can achieve roughly the same spectral efficiency as the traditional OFDM system. 
		
		\item We identify the vulnerability of the synchronization algorithm in the original OFDM system under disguised jamming, and propose a secure synchronization scheme for SP-OFDM which is robust against disguised jamming. In the proposed synchronization scheme, we design an encrypted cyclic prefix (CP) for SP-OFDM, and the synchronization algorithm utilizes the encrypted CP as well as the precoded pilot symbols to estimate time and frequency offsets in the presence of jamming.

		\item We analyze the channel capacity of the traditional OFDM and the proposed SP-OFDM under hostile jamming using the arbitrarily varying channel (AVC) model. It is shown that the deterministic coding capacity of the traditional OFDM is zero under the worst disguised jamming. At the same time, we prove that with the secure randomness shared between the authorized transmitter and receiver, the AVC channel corresponding to SP-OFDM is not symmetrizable, and hence SP-OFDM can achieve a positive capacity under disguised jamming. Note that the authorized user aims to maximize the capacity while the jammer aims to minimize the capacity, we show that the maximin capacity for SP-OFDM under hostile jamming is given by $C=\log \left( 1 + \frac{P_S}{P_J + P_N} \right)$ bits/symbol, where $P_s$ denotes the signal power, $P_J$ the jamming power and $P_N$ the noise power.


	\end{itemize}

	Numerical examples are provided to demonstrate the effectiveness of the proposed system under disguised jamming and channel fading. Potentially, SP-OFDM is a promising modulation scheme for high speed transmission under hostile environments. Moreover, it should be pointed out that the secure precoding scheme proposed in this paper can also be applied to modulation techniques other than OFDM.

	The rest of this paper is organized as follows. The design of the proposed SP-OFDM system is described in Section \ref{Sec:Design}. The synchronization procedure of SP-OFDM is presented in Section \ref{Sec:Sync}. The symmetricity analysis and capacity evaluation of SP-OFDM are presented in Section \ref{Sec:Analysis}. Numerical examples are provided in Section \ref{Sec:Simulation} and we conclude in Section \ref{Sec:Conclude}.

	\section{Secure OFDM System Design under Disguised Jamming} \label{Sec:Design}

	In this section, we introduce the proposed anti-jamming OFDM system with secure precoding and decoding, named as securely procoded OFDM (SP-OFDM). 
	
	\subsection{Transmitter Design with Secure Precoding}\label{Subsec:Precoder}
	The block diagram of the proposed system is shown in Fig. \ref{Fig:Sys}. Let $N_c$ be the number of subcarriers in the OFDM system and $\Phi$ the alphabet of transmitted symbols. For $i = 0, 1, \cdots, N_c-1$ and $k \in \mathbb{Z}$, let $S_{k,i} \in \Phi$ denote the symbol transmitted on the $i$-th carrier of the $k$-th OFDM block\footnote{In literature, the term \emph{OFDM symbol} is often used to denote the symbol block transmitted in one OFDM symbol period. In this paper, to avoid the ambiguity with the data symbols transmitted at each subcarrier, we choose to use the term \emph{OFDM block} instead.}. We denote the symbol vector of the $k$-th OFDM block by $\boldsymbol{S}_k = [S_{k,0}, S_{k,1}, \cdots, S_{k, N_c-1}]^T$. The input data stream is first fed to the channel encoder, mapped to the symbol vector $\boldsymbol{S}_k$, and then fed to the proposed symbol-level secure precoder. 
	
	\begin{figure}[t]
		\centering
		\includegraphics[width=0.9\columnwidth]{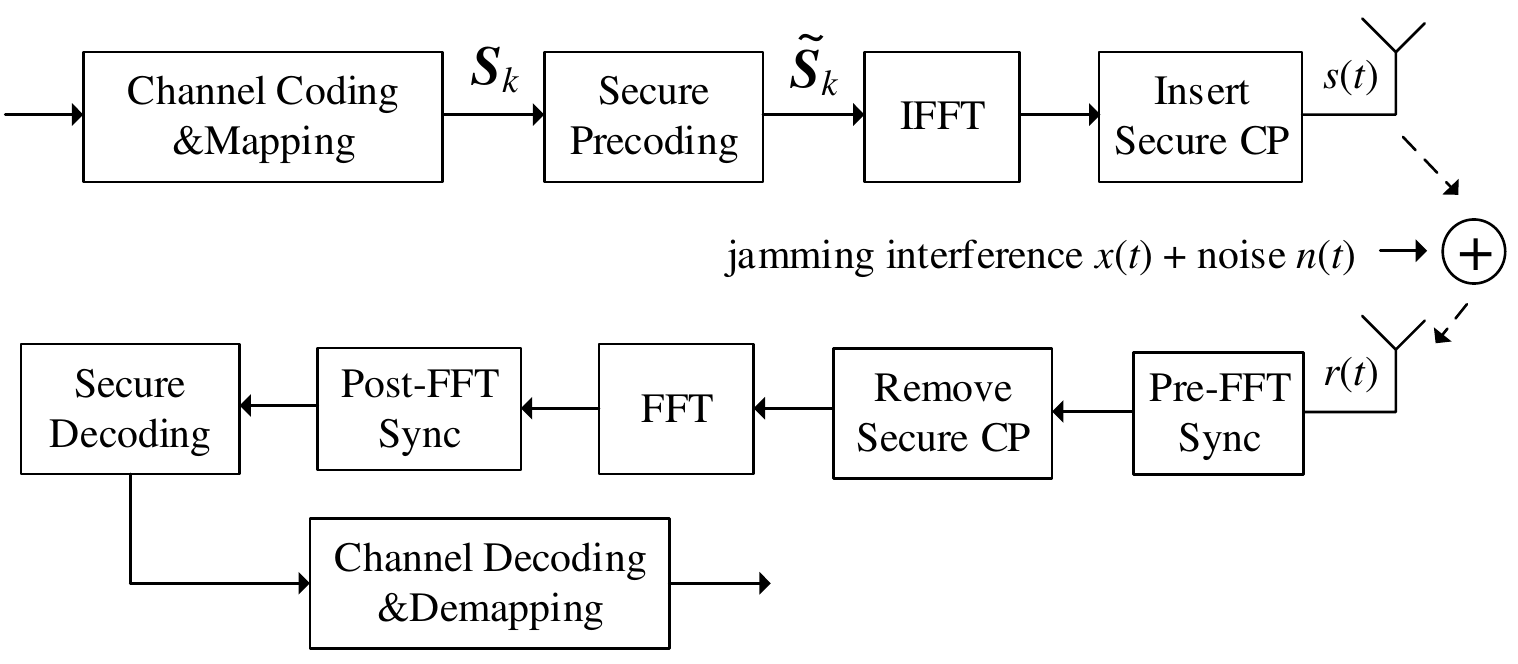}
		\caption{Anti-jamming OFDM design through secure precoding and decoding.}\label{Fig:Sys}
	\end{figure}
	
	As pointed out in \cite{Ahlswede1978, Blackwell1960AMS, Zhang2013TWC1, Zhang2013TWC2}, a key enabling factor for reliable communication under disguised jamming is to introduce shared randomness between the transmitter and receiver, such that the symmetry between the authorized signal and the jamming interference is broken. To maintain full spectral efficiency of the traditional OFDM system, the precoding is performed by multiplying an \emph{invertible} $N_c \times N_c$ precoding matrix $\boldsymbol{P}_k$ to the symbol vector $\boldsymbol{S}_k$, i.e.,
	\begin{equation}
	\tilde{\boldsymbol{S}}_k = \boldsymbol{P}_k \boldsymbol{S}_k.
	\end{equation} 
	
	\begin{figure}[t]
		\centering
		\includegraphics[width=0.9\columnwidth]{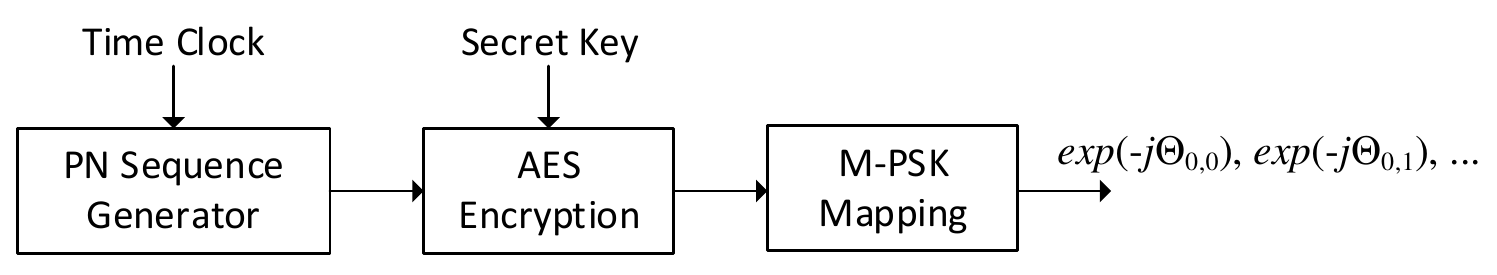}
		\caption{Secure phase shift generator} \label{Fig:AES}
	\end{figure}

	In this paper, we design the precoding matrix $\boldsymbol{P}_k$ to be a diagonal matrix as
	\begin{equation}
	\boldsymbol{P}_k = diag(e^{-j\Theta_{k,0}}, e^{-j\Theta_{k,1}}, \cdots, e^{-j\Theta_{k,N_c - 1}}).
	\end{equation}
	That is, a random phase shift is applied to each transmitted symbol; more specifically, for $i = 0, 1, \cdots, N_c -1$ and $k \in \mathbb{Z}$, a random phase shift $-\Theta_{k,i}$ is applied to the symbol transmitted on the $i$-th carrier of the $k$-th OFDM block. The phase shift changes randomly and independently across sub-carriers and OFDM blocks, and is encrypted so that the jammer has no access to it. More specifically,  $\{\Theta_{k,i}\}$ is generated through a secure phase shift generator as shown in Fig. \ref{Fig:AES}. The secure phase shift generator consists of three parts: (i) a pseudo-noise (PN) sequence generator; (ii) an Advanced Encryption Standard (AES) \cite{AES} encryption module; and (iii) an $M$-PSK mapper. 
	
	The \emph{PN sequence generator} generates a pseudo-random sequence, which is then encrypted with AES. The encrypted sequence is further converted to PSK symbols using an $M$-PSK mapper, where $M$ is a power of $2$, and every $\log_2 M$ bits are converted to a PSK symbol. To facilitate the synchronization process, the PN sequence generator is initialized in the following way: each party is equipped with a global time clock, and the PN sequence generators are reinitialized at fixed intervals. The new state for reinitialization, for example, can be the elapsed time after a specific reference epoch in seconds for the time being, which is public. As the initial state changes with each reinitialization, no repeated PN sequence will be generated.  The security, as well as the randomness of the generated phase shift sequence, are guaranteed by the AES encryption algorithm \cite{AES}, for which the secret encryption key is only shared between the authorized transmitter and receiver.  Hence, the phase shift sequence is random and unaccessible for the jammer. The resulted symbol vector from the secure precoding, $\tilde{\boldsymbol{S}}_k$, is then used to generate the body of OFDM block through IFFT, whose duration is $T_s$. 
		
		In OFDM transceiver design, the synchronization module plays a crucial role: OFDM requires both accurate time and frequency synchronization to avoid inter-symbol interference (ISI) and inter-carrier interference (ICI). In SP-OFDM, we propose a cyclic prefix (CP) based synchronization algorithm, as in traditional OFDM. However, SP-OFDM differs in that its CP is encrypted to ensure the security  under disguised jamming.  
	
	\subsection{Cyclic Prefix Design with Secure Precoding}	
	In traditional OFDM, CP has three major functions: (i) eliminating the ISI between neighboring blocks; (ii) converting the linear convolution of OFDM block body with the channel impulse response into circular convolution under multi-path channel fading; and (iii) eliminating the ICI introduced by multipath propagation. As CP is a copy of the tail of OFDM block body, we can calculate the correlation between CP and the tail of OFDM block to estimate the starting point of each OFDM block \cite{Beek1997TSP} when disguised jamming is absent. 
	
	However, as to be shown in Section \ref{Sec:Sync}, the traditional CP based synchronization is fragile under disguised jamming. As shown in Fig. \ref{Fig:ex}, to ensure the robustness of synchronization, in SP-OFDM, we apply a secure phase shift to part of the CP for each OFDM block. More specifically, the CP of each OFDM block is divided into two parts: for the first part, with a duration of $T_{CP, 1}$, a secure phase shift is applied to the signal. We name this part of CP as CP1; while for the second part, which is of length $T_{CP, 2}$, no special processing is applied. We name the second part as CP2. CP1 is used for effective synchronization under disguised jamming; CP2 maintains the functions of the original CP. To avoid ISI and ICI, both $T_{CP, 1}$ and $T_{CP, 2}$ are chosen to be longer than the maximum delay spread of the channel. 
	
	\begin{figure}
		\centering
		\includegraphics[width=0.9\columnwidth]{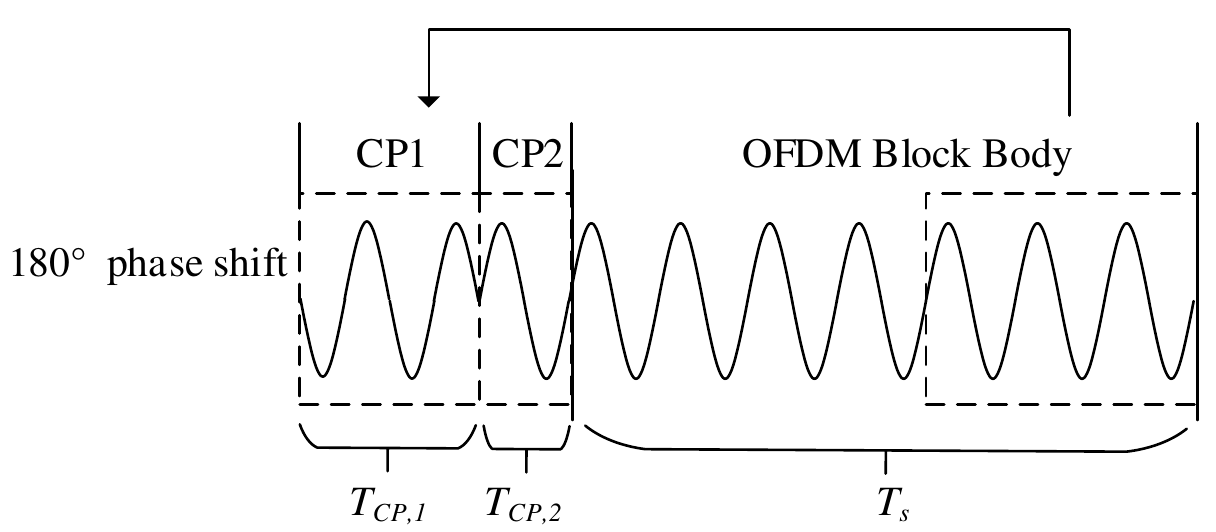}
		\caption{An OFDM waveform example with secure cyclic prefix, illustrated with a $180^\circ$ phase shift on CP1.}\label{Fig:ex}
	\end{figure}

	\begin{figure}
	\centering
	\includegraphics[width=0.9\columnwidth]{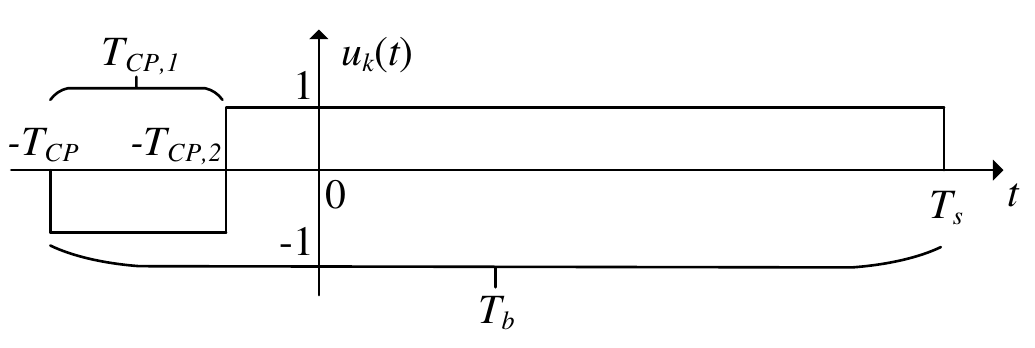}
	\caption{The waveform of $u_k(t)$ with $C_k = -1$.}\label{Fig:ut}
\end{figure}
	
	To ensure the security, the phase shift applied to CP1 is encrypted and varies for each OFDM block. The corresponding secure phase shift sequence can be generated using the same phase shift generator proposed in Fig. \ref{Fig:AES}, with a much lower generation rate, since only one phase shift symbol is needed per OFDM block.	Let $s_k(t)$ denote signal of the $k$-th OFDM block in the time domain by aligning the beginning of the OFDM block body at $t = 0$, and $C_k$ denote the phase shift symbol applied to its CP1; let $u (t)$ be the unit step function, $T_{CP} = T_{CP, 1} + T_{CP, 2}$ and $T_s$ denote the duration of OFDM block body. Define function $u_k(t)$ as
	\begin{equation}
	\small
	u_k(t) \overset{\triangle}{=} C_k [u (t + T_{CP}) - u (t + T_{CP, 2})]  + u (t + T_{CP, 2}) - u (t - T_s).
	\end{equation}
	An example of $u_k(t)$ with $C_k = -1$ is plotted in Fig. \ref{Fig:ut}. For SP-OFDM with secure CP, $s_k(t)$ can be expressed as
	\begin{equation}
	s_k(t) = \frac{1}{N_c}\sum_{i=0}^{N_c - 1} \tilde{S}_{k, i} e^{j\frac{2\pi i}{T_s} t} u_k(t),
	\end{equation}     
	where $\tilde{S}_{k, i} = S_{k, i} e^{-j\Theta_{k,i}}$. Let  $T_b = T_s + T_{CP}$ denote the duration of an OFDM block. Then the entire OFDM signal in the time domain can be expressed as
	\begin{equation}
	s(t) = \sum_{k = -\infty}^{\infty} s_k(t - k T_b).
	\end{equation} 
	
	Even though the receiver can generate identical phase shift sequences used in CP1 generation from the design of Fig. \ref{Fig:AES}, there will still be an offset between the two generated sequences considering the delays in communication and the mismatch between the time clocks. Let $C_k$ and $\tilde{C}_k$ denote the phase shift symbols generated at the transmitter and receiver respectively, and we have
	\begin{equation}
	C_k = \tilde{C}_{k + k_0}, \forall k.
	\end{equation}    
	Since the phase shift sequences are generated from the global time clock, the offset $k_0$ is bounded. The offset $k_0$ can be estimated by the synchronization module at the receiver. Note that synchronization is needed for the precoding matrix sequence $\boldsymbol{P}_k$ as well; for the ease of synchronization, we pair the CP phase shift symbol $C_k$ with the precoding matrix $\boldsymbol{P}_k$ for each OFDM block $k$; that is, for each CP phase shift symbol generated, we generate $N_c$ phase shift symbols in parallel as the sub-carrier phase shifts. In this way, the two phase shift sequences are synchronized, in the sense that once the synchronization on the CP phase shift sequence is obtained, the synchronization on the precoding matrices is achieved automatically.

	\subsection{Receiver Design with Secure Decoding}
	
	We consider an additive white Gaussian noise (AWGN) channel under hostile jamming. The transmitted OFDM signal is subject to an AWGN term, denoted by $n(t)$, and an additive jamming interference $x(t)$. The received OFDM signal can be expressed as
	\begin{equation}
	r(t) = s(t - t_0) e^{j(\omega_0 t + \phi_0)} + x(t) + n(t),
	\end{equation}
	where $t_0$, $\omega_0$ and $\phi_0$ denote the time, frequency and phase offsets between the transmitter and receiver, respectively. Without loss of generality, we can assume that $t_0 \in [0, T_b)$. 
	
	As in the traditional OFDM system, the synchronization module of SP-OFDM consists of two stages: a \emph{pre-FFT synchronization}, which makes use of the correlation between the secure CP and the OFDM body tail to roughly estimate the offsets, and a \emph{post-FFT synchronization}, which makes use of the pilot symbols inserted to certain sub-carriers to obtain a fine estimation. The phase shift offset $k_0$ is also estimated in the pre-FFT stage. The detailed algorithm and analysis on the synchronization of SP-OFDM will be presented in Section \ref{Sec:Sync}.   	
	
	The demodulation module at the receiver will crop the CP to obtain the body of each OFDM block, and apply FFT to obtain the frequency component at each sub-carrier. Under perfect synchronization, the received signal of the $k$-th OFDM block body can be expressed as
	\begin{equation}
		r_k(t) = s_k(t) + x_k(t) + n_k(t),~~ t \in [0, T_s),
	\end{equation}     
	where $x_k(t)$ and $n_k(t)$ are the jamming interference and noise overlaid on the $k$-th OFDM block, respectively. The frequency components of jamming and noise can be calculated as
	\begin{equation}
		J_{k, i} = \sum_{m = 0}^{N_c - 1} x_{k}(\frac{mT_s}{N_c}) e^{-j\frac{2\pi i}{N_c}m}, i = 0, 1, \cdots,  N_c-1,
	\end{equation}
	\begin{equation}
		\bar{N}_{k,i} =  \sum_{m = 0}^{N_c - 1} n_{k}(\frac{mT_s}{N_c}) e^{-j\frac{2\pi i}{N_c}m}, i = 0, 1, \cdots,  N_c-1,
	\end{equation}
	where $\frac{T_s}{N_c}$ is the sampling interval. For an AWGN channel, $\bar{N}_{k,i}$'s are i.i.d. circularly symmetric complex Gaussian random variables with variance $\sigma^2$.	After applying FFT to the received signal, a symbol vector $\tilde{\boldsymbol{R}}_k = [\tilde{R}_{k,0}, \tilde{R}_{k,1}, \cdots, \tilde{R}_{k,N_c-1}]^T$ is obtained for the $k$-th transmitted OFDM block. That is,
	\begin{equation}
	\tilde{\boldsymbol{R}}_k = \boldsymbol{P}_k \boldsymbol{S}_k + \boldsymbol{J}_{k} + \bar{\boldsymbol{N}}_{k}.
	\end{equation}
	where
	\begin{equation}
	\boldsymbol{J}_{k} = [J_{k,0}, J_{k,1}, \cdots, J_{k,N_c-1}]^T,
	\end{equation}
	and
	\begin{equation}
	\bar{\boldsymbol{N}}_{k} = [\bar{N}_{k,0}, \bar{N}_{k,1}, \cdots, \bar{N}_{k,N_c-1}]^T.
	\end{equation}
	
	The secure decoding module multiplies the inverse matrix of $\boldsymbol{P}_k$ to $\tilde{\boldsymbol{R}}_k$, which results in the symbol vector
	\begin{equation}
	\boldsymbol{R}_k = \boldsymbol{S}_k + \boldsymbol{P}_k^{-1} \boldsymbol{J}_{k} + \boldsymbol{P}_k^{-1} \bar{\boldsymbol{N}}_{k}, 
	\end{equation}
	where  $\boldsymbol{R}_k = [R_{k,0}, R_{k,1}, \cdots, R_{k,N_c-1}]^T$, with 	
	\begin{equation}
	R_{k,i} = S_{k, i} + e^{j\Theta_{k,i}} J_{k, i} + N_{k,i}, \label{Eq:chn_model}
	\end{equation}	
	where $N_{k,i} = e^{j\Theta_{k,i}} \bar{N}_{k, i}$, and $\Theta_{k,i}$ is uniformly distributed over $\{ \frac{2\pi i}{M} \mid i =0, 1, \cdots, M-1  \}$.  Note that for any circularly symmetric Gaussian random variable $N$, $e^{j\theta}N$ and $N$ have the same distribution for any angle $\theta$ \cite[p66]{Barry2003}; that is, $N_{k,i}$ is still a circular symmetric complex Gaussian random variable of zero-mean and variance $\sigma^2$. Taking the delay in the communication system into consideration, in this paper, we assume that the authorized user and the jammer do not have pre-knowledge on the sequence of each other. 
		
	
	\section{Synchronization in SP-OFDM under Disguised Jamming}  \label{Sec:Sync}
	In this section, first, we show the vulnerability of the synchronization process in tradition OFDM under disguised jamming attacks; then we propose the synchronization algorithm of SP-OFDM and prove its effectiveness under hostile jamming.
	
	 In modern OFDM systems, there are generally two kinds of approaches to achieve signal synchronization: (i) making use of the correlation between the CP and the tail of each OFDM block \cite{Beek1997TSP}; or (ii) inserting certain training symbols in every OFDM frame \cite{IEEE80211}. However, neither of these two approaches is robust under malicious jamming, especially disguised jamming, where the jammer modulates the inference with OFDM and deceive the receiver into synchronizing with the disguised jamming instead of the legitimate signal. For the training sequence based synchronization approach, even if the training sequence is not public, there is still a chance for the jammer to eavesdrop on the training sequence, and then generate the OFDM modulated disguised jamming with the true training sequence. 
	 
	 \begin{figure}
	 	\centering
	 	\includegraphics[width=0.8\columnwidth]{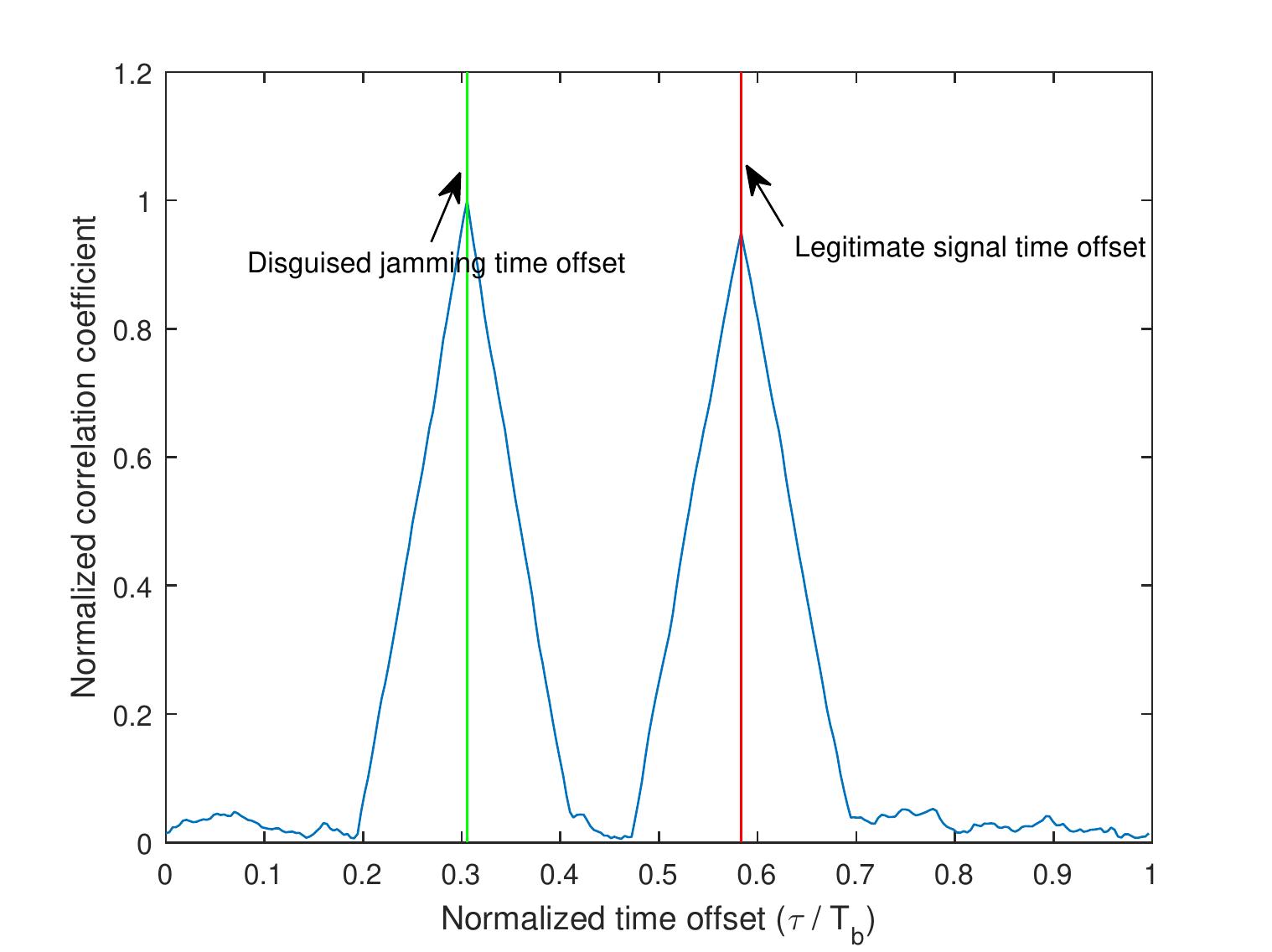}
	 	\caption{Correlation coefficients of the original OFDM under disguised jamming.} \label{Fig:noPrecoding}
	 \end{figure}
	 
	 \textbf{Synchronization of traditional OFDM under disguised jamming:} To demonstrate the damage of disguised jamming, we calculate the CP based correlation coefficients of the traditional OFDM signal  at different time offsets in the AWGN channel under an OFDM modulated disguised jamming. We average the correlation coefficients over multiple OFDM blocks, and the result is shown in Fig. \ref{Fig:noPrecoding}. \emph{Without proper encryption applied to the signal, the legitimate signal and the jamming interference are completely symmetric; we can observe peaks of the correlation coefficients at two different time offsets, one corresponding to that of the legitimate signal and the other corresponding to that of the disguised jamming.} If the jamming power is the same as the signal power, then the probability that the receiver chooses to synchronize with jamming is $50\%$. Obviously, a complete communication failure occurs when the receiver chooses to synchronize with the disguised jamming instead of the legitimate signal.

	 To address this problem, in the synchronization algorithm of SP-OFDM, we apply encrypted phase shifts to the sub-carriers and CP. For the ease of analysis, in the following, we consider an AWGN channel model; the effectiveness of the proposed algorithm in multi-path fading channels will be verified through numerical analysis in Section \ref{Sec:Simulation}. Even though our goal is to guarantee the robustness of SP-OFDM under disguised jamming, in the following analysis, we do not assume any specific form on the jamming interference $x(t)$, that is, we prove the robustness of our algorithm under any form of jamming attacks. Without loss of generality, we denote the combined term of jamming and noise as $z(t) = x(t) + n(t)$, and the received signal can be expressed as
	 \begin{equation}
	 r(t) = s(t - t_0) e^{j(\omega_0 t + \phi_0)} + z(t).
	 \end{equation}
	 
	 \subsection{Pre-FFT Synchronization} 
	 In the pre-FFT stage, we estimate the encrypted phase shift sequence offset $k_0$, time offset $t_0$ and the fractional part of $w_0 T_s / 2\pi$ for frequency offset $w_0$. Since the phase shift sequence $C_k$ is generated from the global time clock, the receiver has rough bounds on $k_0$ relative to the arrival time of the signal. We denote the finite candidate set of offset $k_0$ by $\mathcal{K}$.
	 
	 In the traditional OFDM system, the CP correlation based synchronization algorithm is derived from the maximum-likelihood (ML) rule \cite{Beek1997TSP,Lv2003Globecom}. However, since the jamming distribution is unspecified in our case, the ML rule is not applicable. Instead, we prove the robustness of the synchronization algorithm of SP-OFDM using the Chebychev inequality \cite[Theorem 5.11]{Klenke2008}. 
	 
	 In the pre-FFT stage, the receiver calculates the following correlation coefficient
	 \begin{equation}
	 	Y_k(\tau, d) \overset{\triangle}{=} \int_{\tau - T_{CP} + k T_b}^{\tau - T_{CP,2} + k T_b} r(t) r^*(t + T_s) \tilde{C}^*_{k+d} \text{d}t,~ k \in \mathbb{Z}^*,  \label{Eq:Y} 
	 \end{equation}
	for $\tau \in [0, T_b), d \in \mathcal{K}$. We have the following proposition on $Y_k(\tau, d)$, whose proof is given in the appendix.
	
	\begin{proposition}\label{Prop:pre}
		If the fourth moment of $z(t)$ is bounded for any time instant $t$, i.e., $\mathbb{E}\{ |z(t)|^4 \} < \infty, \forall t \in \mathbb{R}$, then as $K \rightarrow +\infty$, we have
		\begin{equation}
		\small
		\frac{1}{K} \sum_{k = 0}^{K-1} Y_k(\tau, d) = \left\lbrace 
		\begin{array}{rl}
		\frac{P_S}{N_c} v(\tau + T_b - t_0) e^{-j \omega_0 T_s },& d = k_0 - 1, \\
		\frac{P_S}{N_c} v(\tau - t_0) e^{-j \omega_0 T_s },& d = k_0, \\
		\frac{P_S}{N_c} v(\tau - T_b - t_0) e^{-j \omega_0 T_s },& d = k_0 + 1, \\
		0,& otherwise,				 
		\end{array}
		\right.\label{Eq:pre}
		\end{equation}
		 almost surely (a.s.), where
		 	\begin{equation}
		 v(\tau) \overset{\triangle}{=}   \left\lbrace \begin{array}{rl}
		 \tau + T_{CP, 1}  , & -T_{CP, 1} \leq \tau < 0, \\
		 T_{CP, 1} - \tau, & 0 \leq \tau <  T_{CP, 1}, \\
		 0, &  otherwise,
		 \end{array}
		 \right. 
		 \end{equation}
	 and $P_S$ is the average symbol power of constellation $\Phi$.
	\end{proposition} 
	
	 Basing on Proposition \ref{Prop:pre}, to estimate $t_0$ and $k_0$, we search for $\tau$ and $d$ which can maximize $|	\frac{1}{K} \sum_{k = 0}^{K-1} Y_k(\tau, d)|$ for some $K$.  Meanwhile, after we obtain  $t_0$ and $k_0$, the phase of the average correlation coefficient $\frac{1}{K} \sum_{k = 0}^{K-1} Y_k(t_0, k_0)$ is 
	\begin{equation}
	-w_0 T_s \mod 2\pi,
	\end{equation}
	where we can estimate the fractional part of $w_0 T_s/2\pi$ as well. In practice, the jamming interference should be peak power bounded considering the constraints in RF, so we can ensure that the fourth moment of $z(t)$ is bounded. The selection of $K$ depends on the power and the form of the jamming interference. In Section \ref{Sec:Simulation}, we will show that under a disguised jamming, SP-OFDM is able to obtain relatively accurate estimation results with 25 to 30 OFDM blocks. 
	
	As in the traditional OFDM, the CP based synchronization is only able to provide a coarse estimation of time offset $t_0$, especially under multi-path fading, and it requires a fine estimation on the time offset at the post-FFT stage. In addition, from (\ref{Eq:after_sync}), it can be seen that even for a very minor estimation error on the carrier frequency, there still may be an essential phase offset. As long as the range of the time estimation error is smaller than the duration of CP2, without loss of generality, we can model the signal after pre-FFT synchronization as    
	\begin{equation}
	r'(t) = s(t - t_0') e^{j( \frac{2\pi (n_0 + \zeta_0)}{T_s}t + \phi_0)} + z'(t), \label{Eq:after_sync}
	\end{equation} 
	where $z'(t)$ is the jamming interference after pre-FFT synchronization, $t_0' \in [0, T_{CP, 2})$ is the remaining time offset,   $2\pi (n_0 + \zeta_0)/T_s$ is the remaining frequency offset, $n_0$ is an integer and $|\zeta_0| \ll 1$.
	
	\subsection{Post-FFT Synchronization}
	In this stage, we first estimate $n_0 + \zeta_0$  after demodulating the synchronized signal $r'(t)$ in (\ref{Eq:after_sync}) using FFT. Suppose $n_0$ satisfies
	\begin{equation}
	N_l \leq n_0 \leq N_u,
	\end{equation}
	where integers $N_l$ and $N_u$ are determined by the maximal frequency offset between the transmitter and receiver. Basing on (\ref{Eq:after_sync}), to demodulate the $k$-th OFDM block, the receiver applies FFT to signal $r'(t)$ within interval $[kT_b, kT_b + T_s)$. The received signal of $k$-th OFDM block after alignment can be expressed as
	\begin{equation}
	r'_{k}(t) = s_k(t - t_0') e^{j(\frac{2\pi (n_0 + \zeta_0)}{T_s}t + \phi_k)} + z'_{k}(t),~~ t \in [0, T_s),
	\end{equation} 
	where
	\begin{equation}
	\phi_k = \phi_0 + \frac{2\pi (n_0 + \zeta_0) T_b}{T_s}k,
	\end{equation}
	and
	\begin{equation}
	z'_{k}(t) = z'(t + kT_b).
	\end{equation}
	
	Considering the frequency offset $n_0$, the receiver samples the received signal with a sampling frequency $\frac{N_c + N_u - N_l}{T_s}$. Let $N_c' \overset{\triangle}{=} N_c + N_u - N_l$. For $0 \leq i < N_c'$, the FFT applied to $r'_k(t)$  can be expressed as
	\begin{align}
	\!\!R_k(i) &= \sum_{m = 0}^{N_c'-1} r_k'( \frac{mT_s}{N_c'}) e^{-j\frac{2\pi i}{N_c' }m} \nonumber \\ 
	&= \frac{e^{j\phi_k}}{N_c} \sum_{i'=0}^{N_c'-1} \tilde{S}_{k, i'}  \frac{e^{-j\frac{2\pi t_0'}{T_s} i'}(1 - e^{j2\pi \zeta_0})}{1 - e^{j\frac{2\pi (n_0 + \zeta_0 + i' - i)}{N_c'}}} + Z'_{k}(i), \label{Eq:FFT1} 
	\end{align}
	where
	\begin{equation}
	Z'_{k}(i) = \sum_{m = 0}^{N_c'-1} z'_k(\frac{mT_s}{N_c'}) e^{-j\frac{2\pi i}{N_c'}m}~.
	\end{equation}
	Since we assume $|\zeta_0| \ll 1$, for $0 \leq i < N_c'$, we can neglect the ICI in (\ref{Eq:FFT1}) and approximate $R_k(i)$ as
	\begin{equation}
	R_k(i) = \frac{N_c'}{N_c}e^{j\phi_k}  e^{-j\frac{2\pi t_0'}{T_s} [(i-n_0)\!\!\! \mod N_c']} \tilde{S}'_{k, i-n_0} + Z'_{k}(i),
	\end{equation}
	where
	\begin{equation}
	\tilde{S}'_{k, i} = \left\lbrace \begin{array}{rl}
	\tilde{S}_{k,~(i\!\!\! \mod N_c' )},& 0 \leq i\!\!\! \mod N_c' < N_c,\\
	0,& \text{otherwise}.
	\end{array}
	\right. 
	\end{equation} 
	
	The post-FFT synchronization generally utilizes the pilot symbols inserted at certain sub-carriers. For the ease of analysis, we assume a pilot symbol $\boldsymbol{p}$ is placed at sub-carrier $i_p$ of each OFDM block. Note that, as the precoding matrix sequence is synchronized with the CP phase shift sequence, the precoding matrix sequence is synchronized at the receiver after pre-FFT synchronization.  We calculate the following correlation coefficients for each OFDM block $k$:
	\begin{equation}
		\Gamma_{k}(i) \overset{\triangle}{=} R_{k}(i) R_{k+1}^*(i) e^{j(\Theta_{k, i_p} - \Theta_{k+1, i_p})}.
	\end{equation}
	We have the following proposition on $\Gamma_{k}(i)$.
	
	\begin{proposition}\label{Prop:post1}
		If the fourth moment of $z(t)$ is bounded for any time $t$, then as $K \rightarrow +\infty$, we have
	\begin{align}
	&\!\frac{1}{K}\sum_{k = 0}^{K-1} \Gamma_{k}(i) \nonumber \\
	 =& \left\lbrace 
	\begin{array}{rl}
	\!\!\!\left(\frac{N_c'}{N_c}\right)^2\! e^{j \frac{2\pi (n_0 + \zeta_0) T_b}{T_s}} |\boldsymbol{p}|^2,& i = n_0 + i_p\!\!\!\! \mod N_c'~, \\
	0, & \text{otherwise},
	\end{array}
	\right. a.s.. \label{Eq:Gamma_as}
	\end{align} 
	\end{proposition} 
	\begin{proof}
		Note that $\Gamma_{k}(i)$ can be derived as 
	\begin{align}
	&\Gamma_{k}(i)  = [\left({N_c'}/{N_c}\right)^2 e^{j \frac{2\pi (n_0 + \zeta_0) T_b}{T_s}} \tilde{S}'_{k, i-n_0}	 \tilde{S}'^*_{k+1, i-n_0}  \nonumber\\
	& +  \frac{N_c'}{N_c}e^{j\phi_k} \tilde{S}'_{k, i-n_0}	Z'^*_{k+1}(i) + \frac{N_c'}{N_c}e^{j\phi_{k+1}} \tilde{S}'^*_{k+1, i-n_0} Z'_{k}(i)   \nonumber \\
		& + Z'_{k}(i) Z'^*_{k+1}(i) ] 	e^{j(\Theta_{k, i_p} - \Theta_{k+1, i_p})}. \label{Eq:FFT2}
		\end{align}
		Since the phase shifts $\Theta_{k, i}$'s are independent across the sub-carriers, following the approach in the pre-FFT analysis, we have
		\begin{equation}
		\mathbb{E}\{\Gamma_{k}(i) \}\! = \!\left\lbrace 
		\begin{array}{rl}
		\!\!\!\left(\!\frac{N_c'}{N_c}\right)^2\! e^{j \frac{2\pi (n_0 + \zeta_0) T_b}{T_s}} |\boldsymbol{p}|^2,& i = n_0\! +\! i_p \!\!\!\!\mod N_c', \\
		0, & \text{otherwise}.
		\end{array}
		\right. 
		\end{equation}
		while the variance of $\frac{1}{K}\sum_{k = 0}^{K-1} \Gamma_{k}(i)$ converges to $0$ as $K \rightarrow +\infty$. Therefore (\ref{Eq:Gamma_as}) is obtained accordingly. We skip the details here for brevity.
	\end{proof}
	
	Following Proposition \ref{Prop:post1}, $n_0$ can be estimated by finding the $i$ which maximizes $\frac{1}{K}\sum_{k = 0}^{K-1} \Gamma_{k}(i)$. With the $n_0$ obtained, we  can further estimate the frequenecy estimation error $\zeta_0$ in the pre-FFT stage by evaluating the phase of $\frac{1}{K}\sum_{k = 0}^{K-1} \Gamma_{k}((n_0 + i_p)\!\! \mod N_c')$. 
	
	After $n_0$ is estimated, without loss of generality, we can assume $n_0 = 0$ in the following derivation.  In terms of the time offset $t_0'$, given two pilot symbols $\boldsymbol{p}_1$ and $\boldsymbol{p}_2$ located at sub-carriers $i_{p_1}$ and $i_{p_2}$, respectively, we evaluate the following correlation coefficient for each OFDM block $k$:
	\begin{equation}
	\Upsilon_k(i_{p_1}, i_{p_2}) = R_{k}(i_{p_1}) R_{k}^*(i_{p_2}) \boldsymbol{p}_1^* \boldsymbol{p}_2 e^{j(\Theta_{k, i_{p_1}} - \Theta_{k, i_{p_2}})},
	\end{equation}
	and we have the following proposition.
	\begin{proposition}\label{Prop:post2}
		If the fourth moment of $z(t)$ is bounded for any time $t$, then as $K \rightarrow +\infty$, we have
		\begin{equation}
		\frac{1}{K}\!\sum_{k = 0}^{K-1} \Upsilon_k(i_{p_1}, i_{p_2})\! =\!  
		\left(\frac{N_c'}{N_c}\right)^2 \! e^{-j\frac{2\pi t_0'}{T_s} (i_{p_1} - i_{p_2})} |\boldsymbol{p}_1|^2|\boldsymbol{p}_2|^2, 
		\end{equation}
		a.s.
	\end{proposition} 
	The proof of Proposition \ref{Prop:post2} follows a similar approach as Proposition \ref{Prop:pre}, and we skip it for brevity. Note that $t_0' \in [0, T_{CP, 2})$, so $t_0'$ can be estimated from the phase of $\frac{1}{K}\sum_{k = 0}^{K-1} \Upsilon_k(i_{p_1}, i_{p_2})$.	Likewise, the phase offset $\phi_0$ can be estimated by averaging $R_k(i_p) e^{j\Theta_{k, i_p}}$ after compensating for the frequency offset.
	
	\paragraph*{Discussions} 
	Note that under disguised jamming, the estimator averages multiple OFDM blocks to make use of the encrypted signal for an accurate synchronization. In practice, estimation errors always exist in synchronization, so the receiver has to keep track of all the offsets, which can be implemented by the moving average approach. 
	
	The pre-FFT synchronization exploits the correlation between secure CP and the OFDM body tail. The data-aided synchronization approach, i.e., inserting independent training sequence in each OFDM frame, is still an option under disguised jamming if encryption is applied to the training sequence. However, the CP based approach experiences less delay in synchronization. By inserting secure CP for each OFDM block, it is easier to keep track of the time offset continuously.
	
	In the post-FFT stage, inserting more pilots can accelerate the synchronization process; meanwhile, under fading channels, the channel estimation process necessitates pilot symbols over different sub-carrier locations. Channel estimation can be implemented by averaging the received pilot symbols at each sub-carrier location following the approach in synchronization. However, an important point here is that for time varying channels, the duration of the OFDM blocks used for averaging should be smaller than the coherence time so that the channel does not change significantly during each estimation. This is guaranteed in practical systems where the whole OFDM frame duration is shorter than the channel coherence time \cite{IEEE80211}.

	\section{Symmetricity and Capacity Analysis using the AVC Model} \label{Sec:Analysis}
	
	In this section, we analyze the symmetricity and capacity of the proposed SP-OFDM system using the arbitrarily varying channel (AVC) model. Recall that from Section \ref{Sec:Design}, under perfect synchronization, the equivalent channel model of SP-OFDM can be expressed as
	\begin{equation}
	R = S + e^{j \Theta} J + N, \label{Eq:AVC}
	\end{equation}
	where $S \in \Phi, J \in \mathbb{C}$, $N \sim \mathcal{CN}(0, \sigma^2 I)$, $\Theta$ is uniformly distributed over $\{ \frac{2\pi i}{M} \mid i =0, 1, ..., M-1  \}$, and $\mathcal{CN}(\boldsymbol{\mu}, \boldsymbol{\Sigma})$ denotes a circularly symmetric complex Gaussian distribution with mean $\boldsymbol{\mu}$ and variance $\boldsymbol{\Sigma}$.  For generality, in this section, we do not assume any \emph{a priori} information on the jamming $J$, except a finite average power constraint of $P_J$, i.e., $\mathbb{E}\{ |J|^2 \} \leq P_J$. We will show that the AVC corresponding to SP-OFDM is nonsymmetrizable, and hence the AVC capacity of SP-OFDM is positive under disguised jamming.
	


	\subsection{AVC Symmetricity Analysis}
	The arbitrarily varying channel (AVC) model, first introduced in \cite{Blackwell1960AMS}, characterizes the communication channels with unknown states which may vary in arbitrary manners across time. For the jamming channel (\ref{Eq:AVC}) of interest, the jamming symbol $J$ can be viewed as the state of the channel under consideration. The channel capacity of AVC evaluates the data rate of the channel under the most adverse jamming interference among all the possibilities \cite{Csiszar1988TIT}.  Note that unlike the jamming free model where the channel noise sequence is independent of the authorized signal and is independent and identically distributed (i.i.d.), the AVC model considers the possible correlation between the authorized signal and the jamming, as well as the possible temporal correlation among the jamming symbols, which may cause much worse damages to the communication. 
	
	To prove the effectiveness of the proposed SP-OFDM under disguised jamming, we need to introduce some basic concepts and properties of the AVC model. First we revisit the definition of symmetrizable AVC channel.       
	
	\begin{definition}\cite{Csiszar1988TIT}\cite{Csiszar1992TIT} 
		Let $W(\boldsymbol{r} \mid \boldsymbol{s}, \boldsymbol{x})$ denote the conditional PDF of the received signal $R$ given the transmitted symbol $\boldsymbol{s} \in \Phi $ and the jamming symbol $\boldsymbol{x} \in \mathbb{C}$. The AVC channel (\ref{Eq:AVC}) is symmetrizable iff for some auxiliary channel $\pi: \Phi \rightarrow \mathbb{C}$, $\forall \boldsymbol{s}, \boldsymbol{s}' \in \Phi, \boldsymbol{r} \in \mathbb{C}$, we have
		\begin{equation}
			\int_{\mathbb{C}} W(\boldsymbol{r} \mid \boldsymbol{s}, \boldsymbol{x}) \text{d} F_\pi (\boldsymbol{x} | \boldsymbol{s}') = \int_{\mathbb{C}} W(\boldsymbol{r} \mid \boldsymbol{s}', \boldsymbol{x}) \text{d} F_\pi (\boldsymbol{x} | \boldsymbol{s}),  \label{Eq:Symmetric}
		\end{equation}
		 where $F_\pi (\cdot | \cdot)$ is the probability measure of the output of channel $\pi$ given the input, i.e., the conditional CDF 
		\begin{equation}
			F_\pi (\boldsymbol{x} | \boldsymbol{s})  = \Pr\{ Re(\pi(\boldsymbol{s})) \leq Re(\boldsymbol{x}), Im(\pi(\boldsymbol{s})) \leq Im(\boldsymbol{x})  \}, 
		\end{equation}  
		for $\boldsymbol{x} \in \mathbb{C}, \boldsymbol{s} \in \Phi$, where $\pi(\boldsymbol{s})$ denotes the output of channel $\pi$ given input symbol $\boldsymbol{s}$.
	\end{definition}
	We denote the set of all the auxiliary channels, $\pi$'s, that can symmetrize channel (\ref{Eq:AVC}) by $\Pi$, that is,
	\begin{equation}
		\Pi = \left\lbrace \pi \mid \text{ Eq. (\ref{Eq:Symmetric}) is satisfied w.r.t. $\pi$ $ \forall \boldsymbol{s}, \boldsymbol{s}' \in \Phi ,\boldsymbol{r} \in \mathbb{C}$} \right\rbrace .\small
	\end{equation}
	
	 With the average jamming power constraint considered in this paper, we further introduce the definition of $l$-symmetrizable channel.
	
	\begin{definition}\cite{Csiszar1992TIT} \label{Def:l-AVC}
		The AVC channel (\ref{Eq:AVC}) is called $l$-symmetrizable under average jamming power constraint iff there exists a $\pi \in \Pi$ such that
		\begin{equation}
			\int_{\mathbb{C}} |\boldsymbol{x}|^2 \text{d} F_\pi (\boldsymbol{x} | \boldsymbol{s}) < \infty, ~~\forall \boldsymbol{s} \in \Phi. \label{Eq:l-AVC}
		\end{equation}
	\end{definition}
		
	In \cite{Csiszar1992TIT}, it was shown that reliable communication can be achieved as long as the AVC channel is not $l$-symmetrizable.
	
	 \begin{lemma}\cite[Corollary 2]{Csiszar1992TIT} \label{Theorem:AVC}
		 The deterministic coding capacity\footnote{The deterministic coding capacity is defined by the capacity that can be achieved by a communication system, when it applies only one code pattern during the information transmission. In other words, the coding scheme is deterministic and can be readily repeated by other users \cite{Ericson1985TIT}.} of AVC channel (\ref{Eq:AVC}) is positive under any hostile jamming with finite average power constraint iff the AVC is not $l$-symmetrizable. Furthermore, given a specific average jamming power constraint $P_J$, the channel capacity $C$ in this case equals
		 \begin{equation}
		 \begin{array}{c}
		 C = \underset{\mathcal{P}_{S}}{\max} ~\underset{F_J}{\min} ~ I(S, R), \vspace{0.1in}\\
		 s.t.~~ \int_{\mathbb{C}} |\boldsymbol{x}|^2 \text{d} F_J(\boldsymbol{x}) \leq P_J,
		 \end{array}\label{Eq:capacity0}		 	
		 \end{equation}
		 where $I(S, R)$ denotes the mutual information (MI) between the $R$ and $S$ in (\ref{Eq:AVC}), $\mathcal{P}_{S}$ denotes the probability distribution of $S$ over $\Phi$ and $F_J(\cdot)$ the CDF of $J$. 
	 \end{lemma}
	 
	 First, we show that the traditional OFDM system is $l$-symmetrizable under disguised jamming.
	 \begin{theorem}
	 	The traditional OFDM system is $l$-symmetrizable. Therefore, the deterministic coding capacity is zero under the worst disguised jamming with finite average jamming power.
	 \end{theorem}
 	\begin{proof}
 		The AVC model of the traditional OFDM system is
 		\begin{equation}
 			R = S + J + N. \label{Eq:OFDM}
 		\end{equation}
 		We will show that when $S$ and $J$ have the same constellation $\Phi$, hence the same finite average power, the AVC channel is $l$-symmetrizable. It follows from (\ref{Eq:OFDM}) that
 		\begin{equation}
 			W(\boldsymbol{r} \mid \boldsymbol{s}, \boldsymbol{s}') = W(\boldsymbol{r} \mid \boldsymbol{s}', \boldsymbol{s}),~\forall \boldsymbol{s}, \boldsymbol{s}' \in \Phi, \boldsymbol{r} \in \mathbb{C}.
 		\end{equation}
 		Since $\Phi$ has finite average power, the average power constraint (\ref{Eq:l-AVC}) is satisfied by disguised jamming. Hence, channel (\ref{Eq:OFDM}) is $l$-symmetrizable. From Lemma $\ref{Theorem:AVC}$, a necessary condition for a positive  AVC deterministic coding capacity is that the channel is not $l$-symmetrizable. So the traditional OFDM system has zero deterministic coding capacity under disguised jamming with finite average jamming power. 		
 	\end{proof}

	 Next, we show that with the proposed secure precoding, it is impossible to $l$-symmetrize the AVC channel (\ref{Eq:AVC}) corresponding to the SP-OFDM system.
	
	\begin{theorem} \label{Theorem:symm}
		The AVC channel corresponding to the proposed SP-OFDM is not $l$-symmetrizable.
	\end{theorem}
	\begin{proof}
		We prove this result by contradiction. Suppose that there exists a channel $\pi \in \Pi$ such that the AVC channel is $l$-symmetrizable. Denote the output of channel $\pi$ given input $\boldsymbol{x}$ by $\pi(\boldsymbol{x})$, and define the corresponding AVC channel output for inputs $\boldsymbol{s}$ and $\boldsymbol{s}'$ as
		\begin{equation}
			\hat{R}(\boldsymbol{s}, \boldsymbol{s}') = \boldsymbol{s} + \pi(\boldsymbol{s}') e^{j\Theta} + N, 
		\end{equation}
		where $\hat{R}(\boldsymbol{s}, \boldsymbol{s}')$ denotes the channel output. Following (\ref{Eq:Symmetric}),  $\hat{R}(\boldsymbol{s}, \boldsymbol{s}')$ and $\hat{R}(\boldsymbol{s}', \boldsymbol{s})$ have the same distribution. Let ${\text{\large$\varphi$}}_{X}(\omega_1, \omega_2)$ denote the characteristic function (CF) of a complex random variable $X$. So we have
		\begin{equation}
			{\text{\large$\varphi$}}_{\hat{R}(\boldsymbol{s}, \boldsymbol{s}')}(\omega_1, \omega_2) \equiv {\text{\large$\varphi$}}_{\hat{R}(\boldsymbol{s}', \boldsymbol{s})}(\omega_1, \omega_2), \label{Eq: tmp1}
		\end{equation}
		and
		\begin{equation}
			{\text{\large$\varphi$}}_{\hat{R}(\boldsymbol{s}, \boldsymbol{s}')}(\omega_1, \omega_2) = {\text{\large$\varphi$}}_{[\boldsymbol{s} + \pi(\boldsymbol{s}') e^{j\Theta}]}(\omega_1, \omega_2) ~{\text{\large$\varphi$}}_{N}(\omega_1, \omega_2),
		\end{equation}
		where, for the complex Gaussian noise $N$, we have
		\begin{equation}
			{\text{\large$\varphi$}}_{N}(\omega_1, \omega_2) = e^{-\frac{\sigma^2}{4}(w_1^2 + w_2^2)},~~\omega_1, \omega_2 \in (-\infty, +\infty),
		\end{equation}
		which is non-zero over $\mathbb{R}^2$. Thus by eliminating the characteristic functions of the Gaussian noises on both sides of equation (\ref{Eq: tmp1}), we have
		\begin{equation}
			{\text{\large$\varphi$}}_{[\boldsymbol{s} + \pi(\boldsymbol{s}') e^{j\Theta}]}(\omega_1, \omega_2)\! =\! {\text{\large$\varphi$}}_{[\boldsymbol{s}' + \pi(\boldsymbol{s}) e^{j\Theta}]}(\omega_1, \omega_2). \label{Eq:tmp2}
		\end{equation}  
		 for $\omega_1, \omega_2 \in (-\infty, +\infty)$. Let $\boldsymbol{s} = s_1 + js_2$, we can then express ${\text{\large$\varphi$}}_{[\boldsymbol{s} + \pi(\boldsymbol{s}') e^{j\Theta}]}(\omega_1, \omega_2)$ as
		 \begin{equation}
		 	{\text{\large$\varphi$}}_{[\boldsymbol{s} + \pi(\boldsymbol{s}') e^{j\Theta}]}(\omega_1, \omega_2) = e^{js_1\omega_1 + js_2\omega_2} {\text{\large$\varphi$}}_{[\pi(\boldsymbol{s}') e^{j\Theta}]}(\omega_1, \omega_2), \label{Eq:tmp3}
		 \end{equation}
		and
		\begin{eqnarray}
		 	&&{\text{\large$\varphi$}}_{[\pi(\boldsymbol{s}') e^{j\Theta}]}(\omega_1, \omega_2) \nonumber \\
		 	&=&\mathbb{E} \{ e^{j\omega_1 Re(\pi(\boldsymbol{s}') e^{j\Theta}) + j\omega_2 Im(\pi(\boldsymbol{s}') e^{j\Theta})} \} \nonumber \\
		 	&=& \int_{\mathbb{C}} \mathbb{E} \{ e^{j\omega_1 Re(\boldsymbol{x} e^{j\Theta}) + j\omega_2 Im(\boldsymbol{x} e^{j\Theta})} \} \text{d} F_\pi (\boldsymbol{x} | \boldsymbol{s}'). 
		 \end{eqnarray}
		 Recall that under the proposed secure precoding scheme, $\Theta$ is uniformly distributed over $\{ \frac{2\pi i}{M} \mid i =0, 1, ..., M-1  \}$, where $M$ is a power of $2$. We have
		 \begin{align}
		 	&\mathbb{E} \{ e^{j\omega_1 Re(\boldsymbol{x} e^{j\Theta}) + j\omega_2 Im(\boldsymbol{x} e^{j\Theta})} \}  \nonumber \\
		 	=& \frac{1}{M} \sum_{i = 0}^{M-1} e^{j \omega_1 |\boldsymbol{x}| \cos(\frac{2\pi i}{M} +\arg(\boldsymbol{x})) + j \omega_2 |\boldsymbol{x}| \sin (\frac{2\pi i}{M} +\arg(\boldsymbol{x}))} \nonumber \\
		 	=&\frac{2}{M} \sum_{i = 0}^{M/2-1} \cos \left\lbrace   \omega_1 |\boldsymbol{x}|\cos[{2\pi i}/{M} +\arg(\boldsymbol{x})] \right. \nonumber \\
		 	&~~~~~~~~~~~~~\left. +   \omega_2 |\boldsymbol{x}| \sin [{2\pi i}/{M} +\arg(\boldsymbol{x})]  \right\rbrace , \label{Eq:tmp4}
		 \end{align}
		 which is of real value for $\omega_1, \omega_2 \in (-\infty, +\infty)$. So ${\text{\large$\varphi$}}_{[\pi(\boldsymbol{s}') e^{j\Theta}]}(\omega_1, \omega_2)$ and ${\text{\large$\varphi$}}_{[\pi(\boldsymbol{s}) e^{j\Theta}]}(\omega_1, \omega_2)$ are also real-valued over $\mathbb{R}^2$. For $\boldsymbol{s} \neq \boldsymbol{s}'$ and $\boldsymbol{s}' = s_1' + js_2'$,  $e^{j[(s_1 - s_1')\omega_1 + (s_2 - s_2')\omega_2]}$ has non-zero imaginary part for $(s_1 - s_1')\omega_1 + (s_2 - s_2')\omega_2 \neq n\pi$, $n\in \mathbb{Z}$. Without loss of generality, we assume $s_1 \neq s_1'$. From (\ref{Eq:tmp2}), (\ref{Eq:tmp3}) and (\ref{Eq:tmp4}), for $\omega_1 + \frac{s_2 - s_2'}{s_1 - s_1'} \omega_2 \neq \frac{n\pi}{s_1 - s_1'} , \forall n\in \mathbb{Z}$, we have
		 \begin{equation}
		 	{\text{\large$\varphi$}}_{[\pi(\boldsymbol{s}) e^{j\Theta}]}(\omega_1, \omega_2) = 0. \label{Eq:CF_zero1}
		 \end{equation}
			 
		On the other hand, the characteristic function of an RV should be uniformly continuous in the real domain \cite[Theorem 15.21]{Klenke2008}. So for any  fixed $\omega_2 \in (-\infty, \infty)$, we should have
		\begin{align}
			&{\text{\large$\varphi$}}_{[\pi(\boldsymbol{s}) e^{j\Theta}]}( \frac{n\pi - (s_2 - s_2')\omega_2}{s_1 - s_1'} , \omega_2) \nonumber \\ 
			=& \lim\limits_{ {\omega_1 \rightarrow \frac{n\pi - (s_2 - s_2')\omega_2}{s_1 - s_1'}}} {\text{\large$\varphi$}}_{[\pi(\boldsymbol{s}) e^{j\Theta}]}(\omega_1, \omega_2),~\forall n\in \mathbb{Z}.
		\end{align} 
		   For $\omega_1 \in \left(\frac{(n-1)\pi - (s_2 - s_2')\omega_2}{s_1 - s_1'},\frac{n\pi - (s_2 - s_2')\omega_2}{s_1 - s_1'}\right)\cup \left(\frac{n\pi - (s_2 - s_2')\omega_2}{s_1 - s_1'}, \frac{(n+1)\pi - (s_2 - s_2')\omega_2}{s_1 - s_1'}\right)$, ${\text{\large$\varphi$}}_{[\pi(\boldsymbol{s}) e^{j\Theta}]}(\omega_1, \omega_2) \equiv 0$, so
		   \begin{equation}
		   	{\text{\large$\varphi$}}_{[\pi(\boldsymbol{s}) e^{j\Theta}]}( \frac{n\pi - (s_2 - s_2')\omega_2}{s_1 - s_1'} , \omega_2) = 0,~\forall n\in \mathbb{Z}. \label{Eq:CF_zero2}
		   \end{equation}
		 Combining (\ref{Eq:CF_zero1}) and (\ref{Eq:CF_zero2}), we have
		 \begin{equation}
		 	{\text{\large$\varphi$}}_{[\pi(\boldsymbol{s}) e^{j\Theta}]}(\omega_1, \omega_2) = 0, ~~\forall\omega_1, \omega_2 \in (-\infty, \infty). \label{Eq:zero3}
		 \end{equation}
		 However, (\ref{Eq:zero3}) cannot be a valid characteristic function for any RV. Therefore, the auxiliary channel $\pi$ does not exist, and $\Pi$ is empty. Hence, the AVC channel is not $l$-symmerizable.
	\end{proof}
	Following Lemma \ref{Theorem:AVC}, the result in Theorem \ref{Theorem:symm} implies that the proposed SP-OFDM will always have positive capacity under any hostile jamming with finite average power constraint.	The next subsection is focused on how to calculate the channel capacity of SP-OFDM under hostile jamming.
	
	\subsection{Capacity Analysis}
	From Lemma \ref{Theorem:AVC}, the capacity of channel $R = S + e^{j \Theta} J + N$ is given by
	\begin{equation}
	\begin{array}{c}
	C = \underset{\mathcal{P}_{S}}{\max} ~\underset{F_J}{\min} ~ I(S, R), \vspace{0.1in}\\
	s.t.~~ \int_{\mathbb{C}} |\boldsymbol{x}|^2 \text{d} F_J(\boldsymbol{x}) \leq P_J.
	\end{array} \nonumber	 	
	\end{equation}
	It is hard to obtain a closed form solution of the channel capacity for a general discrete transmission alphabet $\Phi$. However, if we relax the distribution of the transmitted symbol $S$ from the discrete set $\Phi$ to the entire complex plane $\mathbb{C}$ under an average power constraint, we are able to obtain the following result on channel capacity.	
	\begin{theorem}\label{Theorem:capacity_thumb}
		The deterministic coding capacity of SP-OFDM is positive under any hostile jamming. More specifically, let the alphabet $\Phi = \mathbb{C}$ and the average power of $S$ being upper bounded by $P_S$, then the maximin channel capacity in (\ref{Eq:capacity0}) under average jamming power constraint $P_J$ and noise power $P_N = \sigma^2$ is given by
		\begin{equation}
			C = \log \left( 1 + \frac{P_S}{P_J + P_N}\right). \label{Eq:capacity_gauss}
		\end{equation}
		The capacity is achieved at input distribution $\mathcal{CN}(0, P_S)$ and jamming distribution $\mathcal{CN}(0, P_J)$.
	\end{theorem}
	To prove Theorem \ref{Theorem:capacity_thumb}, we need the following lemma \cite[Lemma 4]{Csiszar1992TIT}.
	\begin{lemma}\label{Lemma:convexconcave}
		Mutual information $I(S, R)$ is concave with respect to the input distribution $F_S(\cdot)$ and convex with respect to the jamming distribution $F_J(\cdot)$.
	\end{lemma}    
	\begin{proof}[Proof of Theorem \ref{Theorem:capacity_thumb}]
		First, following Lemma \ref{Theorem:AVC} and Theorem \ref{Theorem:symm}, we can get that the deterministic coding capacity of SP-OFDM is positive under any hostile jamming. 
		
		Second, we will evaluate the channel capacity of SP-OFDM under hostile jamming. When the support of $S$ is $\Phi = \mathbb{C}$, the whole complex plane, following Lemma \ref{Theorem:AVC}, the channel capacity in (\ref{Eq:capacity0}) equals
		\begin{eqnarray}
				 &C = \underset{F_S}{\max} ~\underset{F_J}{\min} ~ I(S, R), \vspace{0.1in} \label{Eq:capacity1}\\
				 s.t.& \int_{\mathbb{C}} |\boldsymbol{x}|^2 \text{d} F_S(\boldsymbol{x}) \leq P_S, \label{Eq:input_constraint} \\
				 & \int_{\mathbb{C}} |\boldsymbol{x}|^2 \text{d} F_J(\boldsymbol{x}) \leq P_J, \label{Eq:jamming_constraint}
		\end{eqnarray}
		where $F_S(\cdot)$ denotes the CDF function of $S$ defined on $\mathbb{C}$, and (\ref{Eq:input_constraint}) and (\ref{Eq:jamming_constraint}) denote the average power constraints on the input and the jamming, respectively.
		
		 We denote the $I(S, R)$ w.r.t the input distribution $F_S(\cdot)$ and the jamming distribution $F_J(\cdot)$ by $\phi(F_S, F_J)$.  Following Lemma \ref{Lemma:convexconcave}, $\phi(F_S, F_J)$ is concave w.r.t. $F_S(\cdot)$ and convex w.r.t. $F_J(\cdot)$. As shown in \cite{Borden1985JCO}, if we can find the input distribution $F_S^*$ and the jamming distribution $F_J^*$ such that
		\begin{equation}
		\phi (F_S, F_J^*) \leq \phi (F_S^*, F_J^*) \leq \phi (F_S^*, F_J),
		\end{equation}
		for any $F_S$ and $F_J$ satisfying the average power constraints (\ref{Eq:input_constraint}) and (\ref{Eq:jamming_constraint}), respectively, then
		\begin{equation}
		\phi (F_S^*, F_J^*) = C.
		\end{equation}
		That is, the pair $(F_S^*, F_J^*)$ is  the saddle point of the max-min problem  in equation (\ref{Eq:capacity1}) \cite{Du1995}.
		
		Assume the jamming interference is circularly symmetric complex Gaussian with average power $P_J$, that is, $F_J^* = \mathcal{CN}(0, P_J)$. Note that the phase shift would not change the distribution of a complex Gaussian RV, and the fact that the jamming $J$ and the noise $N$ are independent, hence the jammed channel in this case is equivalent to a complex AWGN channel with noise power $P_J + P_N$, where the capacity achieving input distribution is also a complex Gaussian with power $P_S$, that is, $F_S^* = \mathcal{CN}(0, P_S)$. It follows that for any input distribution $F_S$ satisfying the power constraint $P_S$,
		\begin{equation}
			\phi (F_S, \mathcal{CN}(0, P_J)) \leq \phi (\mathcal{CN}(0, P_S), \mathcal{CN}(0, P_J)).
		\end{equation}
		
		On the other hand, when the input distribution is $F_S^* = \mathcal{CN}(0, P_S)$, the worst noise in terms of capacity for Gaussian input is Gaussian \cite{Gamal2012}. Since $e^{j \Theta} J + N$ is complex Gaussian with power $P_J + P_N$ if $F_J^* = \mathcal{CN}(0, P_J)$,  then for any jamming distribution $F_J$ satisfying the power constraint $P_J$, 
		\begin{equation}
		\phi (\mathcal{CN}(0, P_S), \mathcal{CN}(0, P_J)) \leq \phi (\mathcal{CN}(0, P_S), F_J).
		\end{equation}  
		
		 So the saddle point $(F_S^*, F_J^*)$ is achieved at $(\mathcal{CN}(0, P_S)$, $\mathcal{CN}(0, P_J))$, where the corresponding channel capacity is
		\begin{equation}
			C = \log \left( 1 + \frac{P_S}{P_J + P_N}\right),
		\end{equation}
		which completes the proof.
	\end{proof}

	\begin{table*}
	\scriptsize
	\caption{SP-OFDM parameters in numerical results ($T_s$: duration of OFDM body)}\label{Tab:params}
	\centering
	\begin{tabular}{c|c||c|c||c|c}
		\hline
		\hline
		Carrier number $N_c$ & 128 & CP1 duration $T_{CP, 1}$ & $T_s/8$ & CP2 duration $T_{CP,2}$ & $T_s/16$ \\
		Number of candidate phase shift offset $|\mathcal{K}|$ & 50 & Signal-to-noise ratio (dB) & 15 & Phase shift constellation size $M$ & 16 \\
		\hline 
	\end{tabular}
\end{table*}	

	\section{Numerical Results} \label{Sec:Simulation}
	In this section, we evaluate the synchronization and bit error rate (BER) performances of the proposed SP-OFDM system under disguised jamming attacks through numerical examples. Throughout this section, we consider the case where the malicious user generates disguised jamming using OFDM, with the same format and power level as that of the legitimate signal.

		\begin{figure}
		\centering
		\includegraphics[width=0.9\columnwidth]{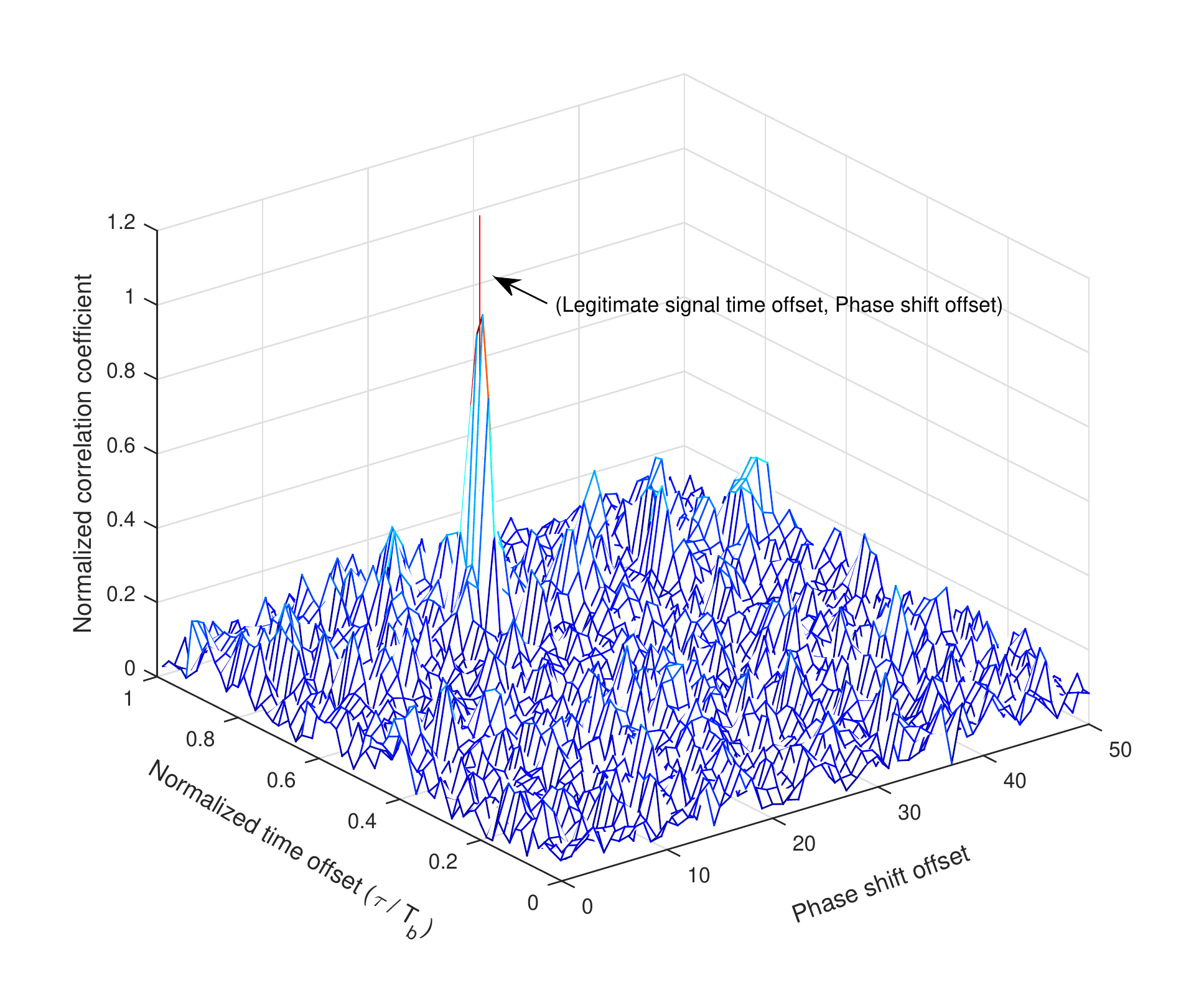}
		\caption{Correlation coefficients of SP-OFDM at different time and phase shift sequence offsets under disguised jamming.} \label{Fig:timeAWGN}
	\end{figure}
	
		\begin{figure}
		\centering
		\begin{subfigure}{.5\columnwidth}
			\centering
			\includegraphics[width=\columnwidth]{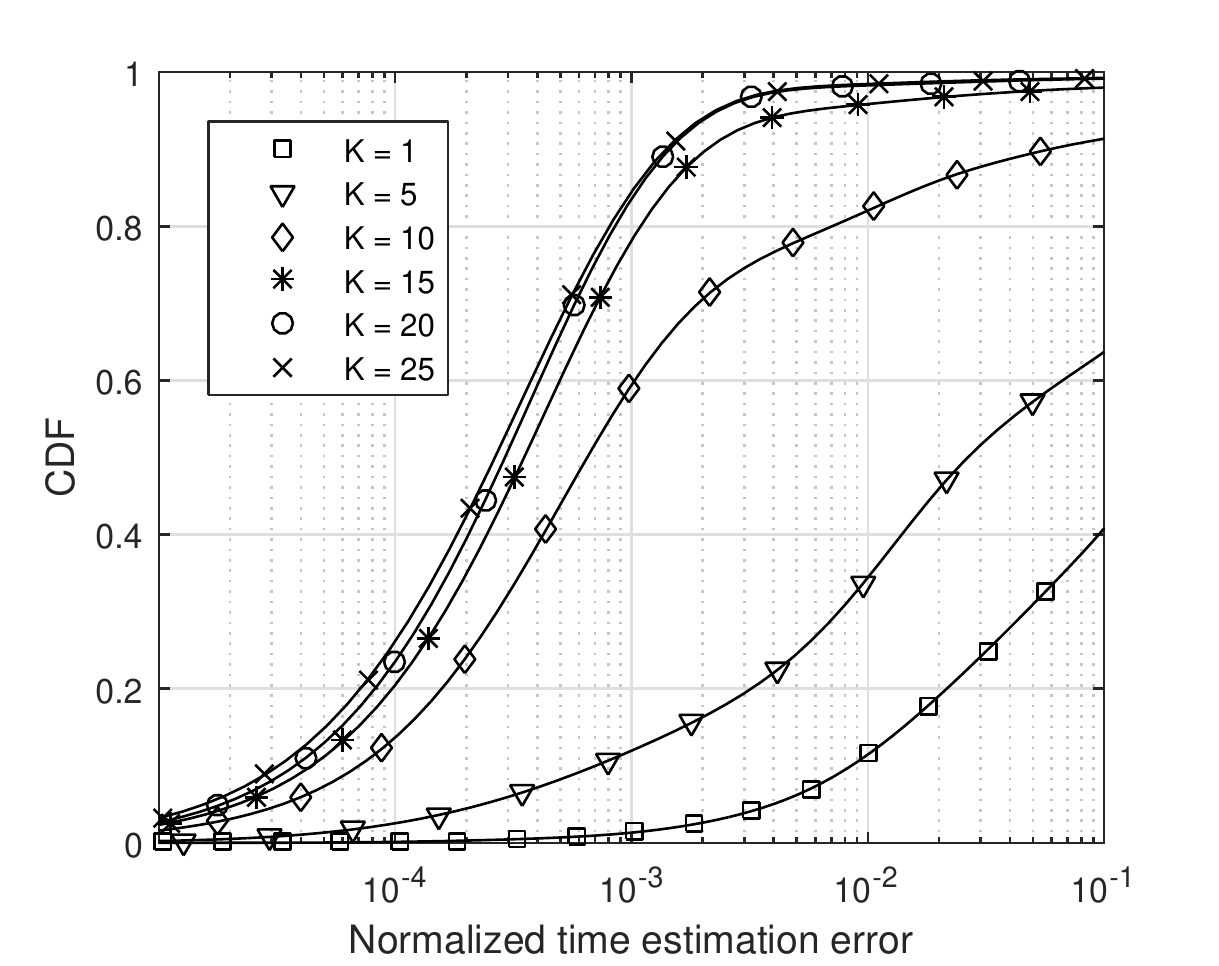}
		\end{subfigure}%
		\begin{subfigure}{.5\columnwidth}
			\centering
			\includegraphics[width=\columnwidth]{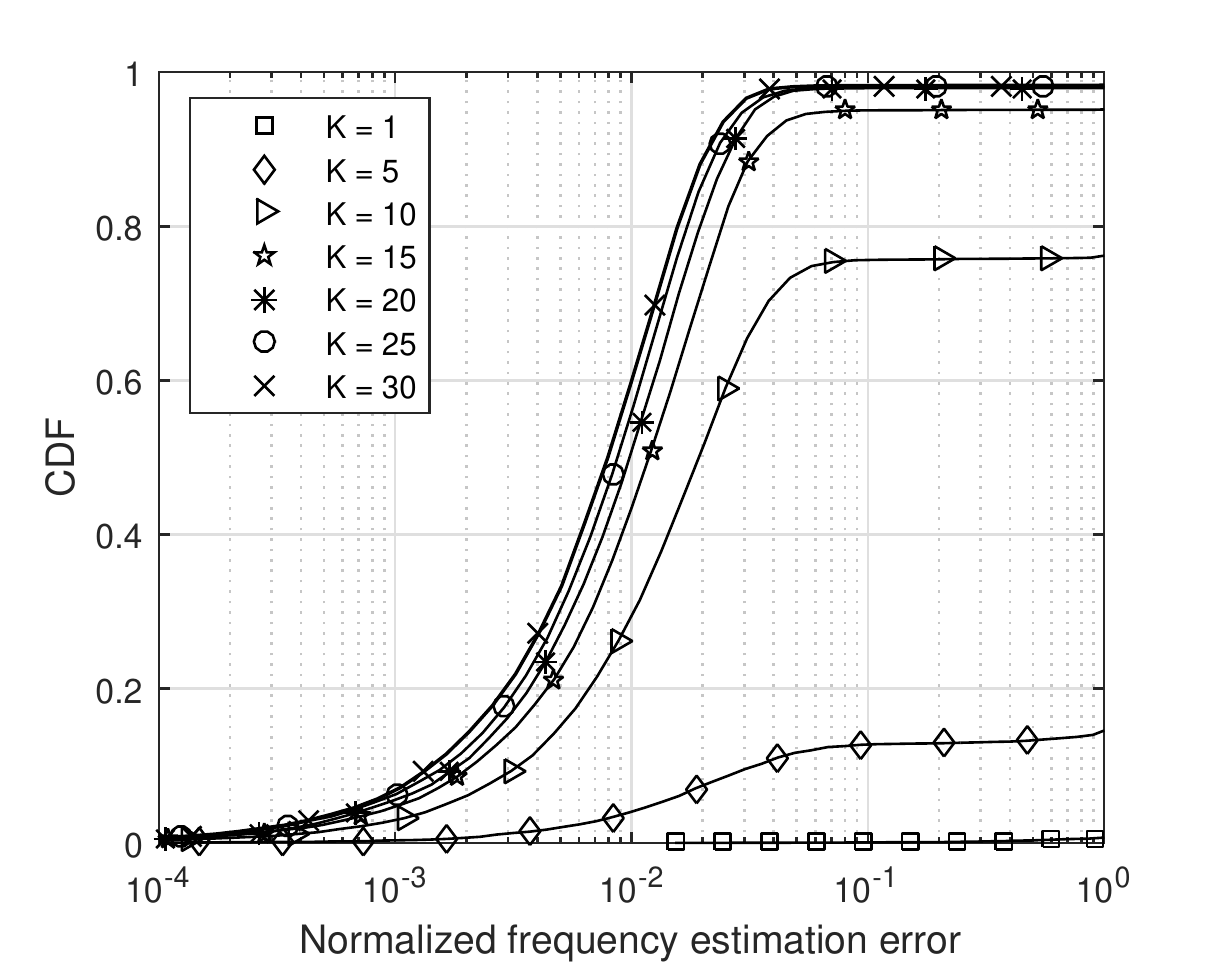}
		\end{subfigure}
		\caption{The synchronization error distribution under AWGN channels with disguised jamming attack.}
		\label{Fig:errorAWGN}
	\end{figure}
		
	\textbf{Example 1: Synchronization performance under disguised jamming in AWGN channels:} In this example, we verify the robustness of SP-OFDM under disguised jamming in terms of synchronization for AWGN channels. The system parameters are listed in Table \ref{Tab:params}. We first compute the average correlation coefficients at different time offsets and phase shift sequence offsets for the received signal as in (\ref{Eq:Y}), and the result is plotted in Fig. \ref{Fig:timeAWGN} for $K=40$\footnote{In the 802.11a WLAN \cite{IEEE80211}, 40 OFDM blocks correspond to 1440 data bytes with 64QAM mapping,  while the OFDM frame length can be as large as 2312 bytes.}. Here, $K$ denotes the number of OFDM blocks used for estimation. It shows that with the secure precoding scheme, even under disguised jamming, the receiver is able to correctly estimate the time offset as well as the phase shift sequence offset of the legitimate signal. Then we simulate the synchronization accuracy of SP-OFDM by calculating the cumulative distribution functions (CDFs) of the estimation errors with different numbers of OFDM blocks $K$ to average the correlation coefficients. We normalize the time offset by the duration of one OFDM block $T_b$ and the frequency offset by the sub-carrier spacing $1/T_s$, and the results are shown in Fig. \ref{Fig:errorAWGN}.  It can be observed that under the given setup, with $25$ OFDM blocks to compute the correlation coefficients, the synchronization algorithm is robust under disguised jamming, where $99\%$ cases have less than $0.01$ \emph{normalized} time offset estimation errors and $98\%$ cases have less than $0.04$ \emph{normalized} frequency offset estimation errors.    
	
		\begin{figure}
		\centering
		\begin{subfigure}{.5\columnwidth}
			\centering
			\includegraphics[width=\columnwidth]{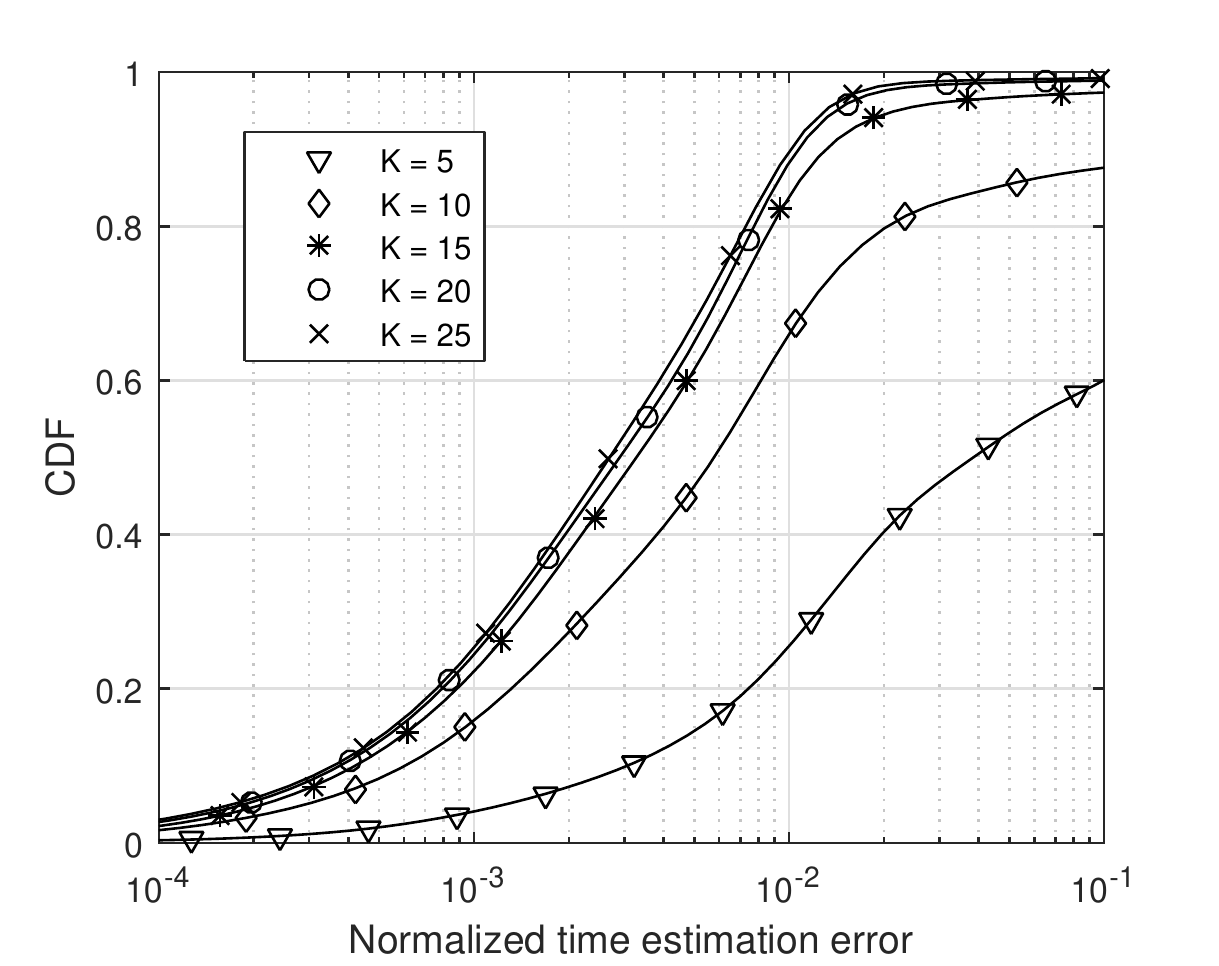}
		\end{subfigure}%
		\begin{subfigure}{.5\columnwidth}
			\centering
			\includegraphics[width=\columnwidth]{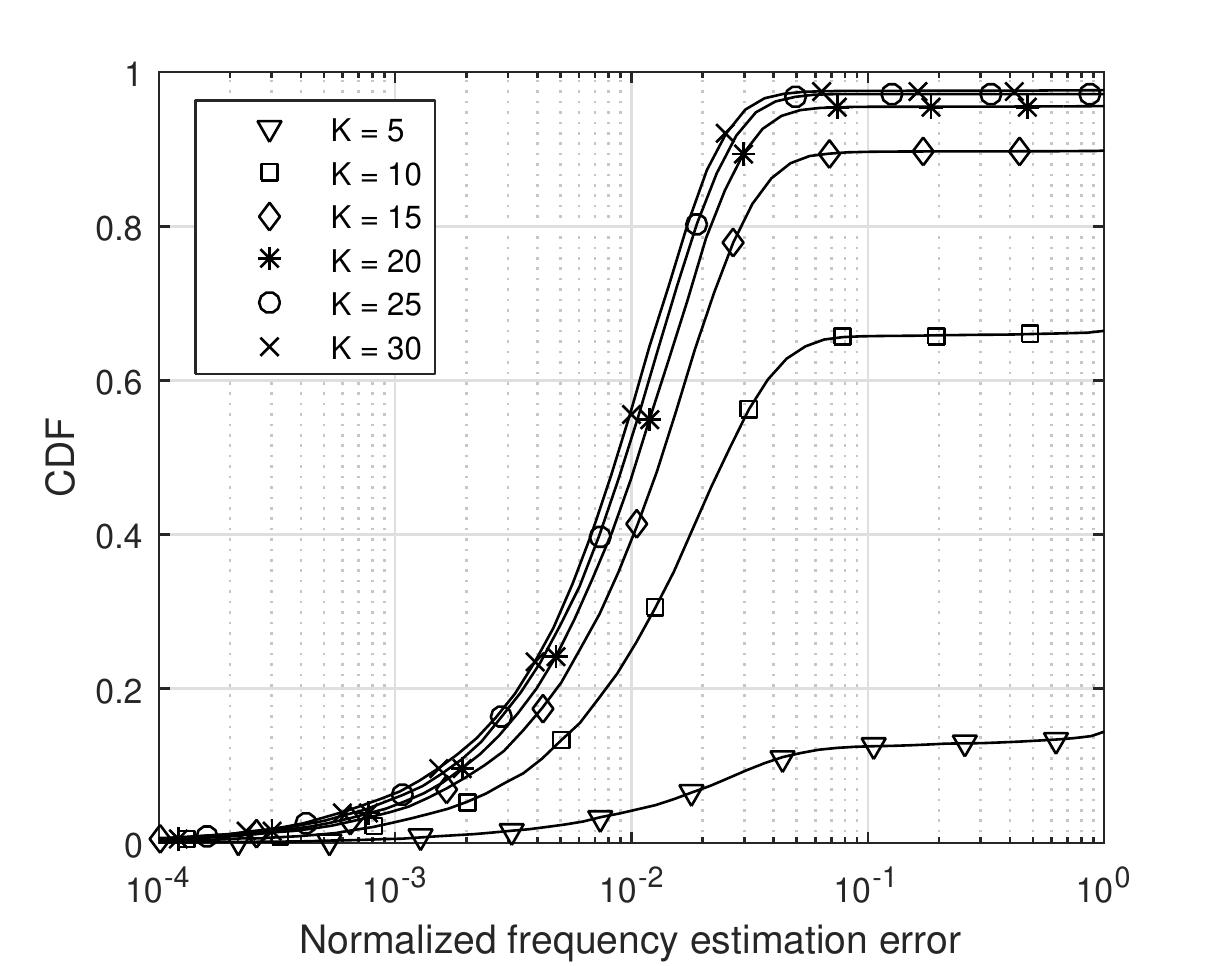}
		\end{subfigure}
		\caption{The synchronization error distribution under static multi-path fading channels with disguised jamming attack.}
		\label{Fig:errorMP}
	\end{figure}
	
	\begin{figure}
		\centering
		\begin{subfigure}{.5\columnwidth}
			\centering
			\includegraphics[width=\columnwidth]{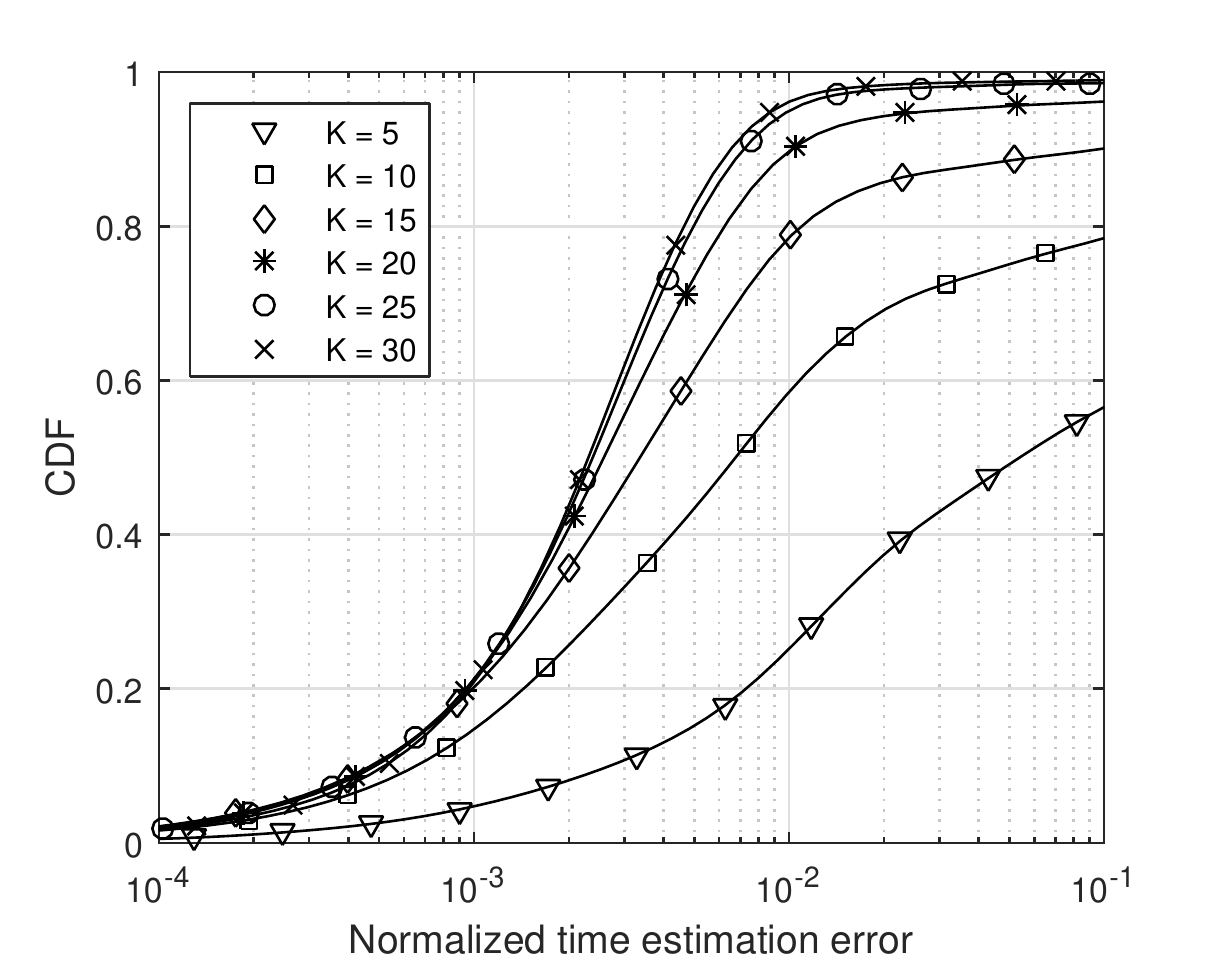}
		\end{subfigure}%
		\begin{subfigure}{.5\columnwidth}
			\centering
			\includegraphics[width=\columnwidth]{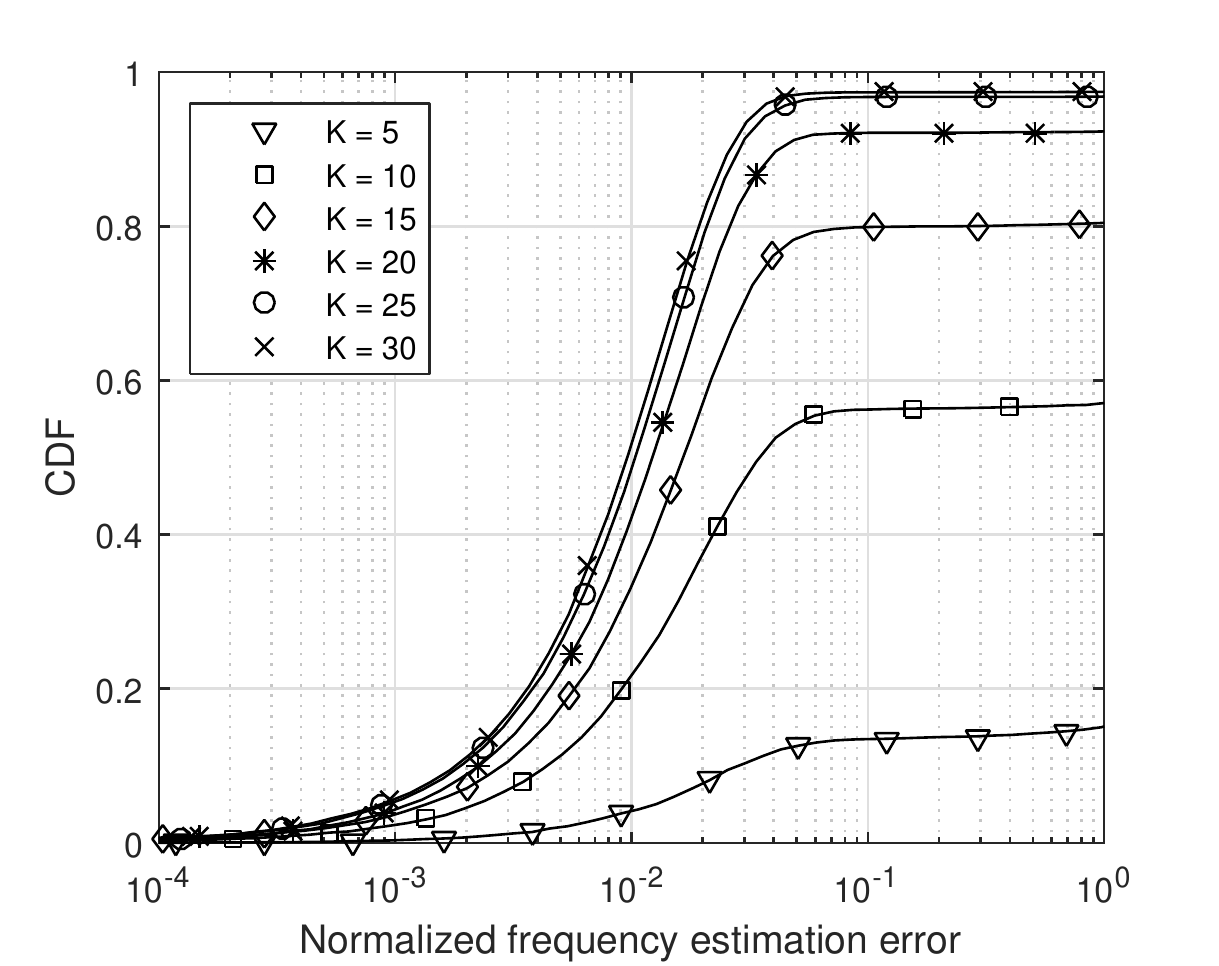}
		\end{subfigure}
		\caption{The synchronization error distribution under time varying multi-path fading channels with disguised jamming attack.}
		\label{Fig:errorMPD}
	\end{figure}

	\textbf{Example 2: Synchronization performance under disguised jamming in multi-path fading channels:} In this example, we simulate the synchronization accuracy of SP-OFDM under disguised jamming in static and time varying multi-path fading channels, which are modeled as 4 paths fading channels with a maximum delay spread of $3T_s/256$.  Fig. \ref{Fig:errorMP} shows the estimation error distribution in the static channel.  A slight performance loss is observed compared with the AWGN case, where $98\%$ cases have  less than $0.02$ \emph{normalized} time offset estimation errors and $96.5\%$ cases have less than $0.04$ \emph{normalized} frequency offset estimation errors using $25$ OFDM blocks in estimation. To demonstrate the effectiveness of the synchronization algorithm under slow time varying channels, we introduce a Doppler shift to each path with a maximum value of $2\%$ sub-carrier spacing ($0.02/T_s$) in the multi-path fading channel. Fig. \ref{Fig:errorMPD} shows the estimation error distribution under the time-varying multi-path fading channel, where around $98\%$ cases have  less than $0.02$ normalized time offset estimation errors and $96.5\%$ cases have less than $0.04$ normalized frequency offset estimation errors using $30$ OFDM blocks in estimation. The simulation results illustrate the robustness of SP-OFDM against disguised jamming attacks under various channel conditions.

	\begin{figure}[t]
		\centering
		\includegraphics[width=\columnwidth]{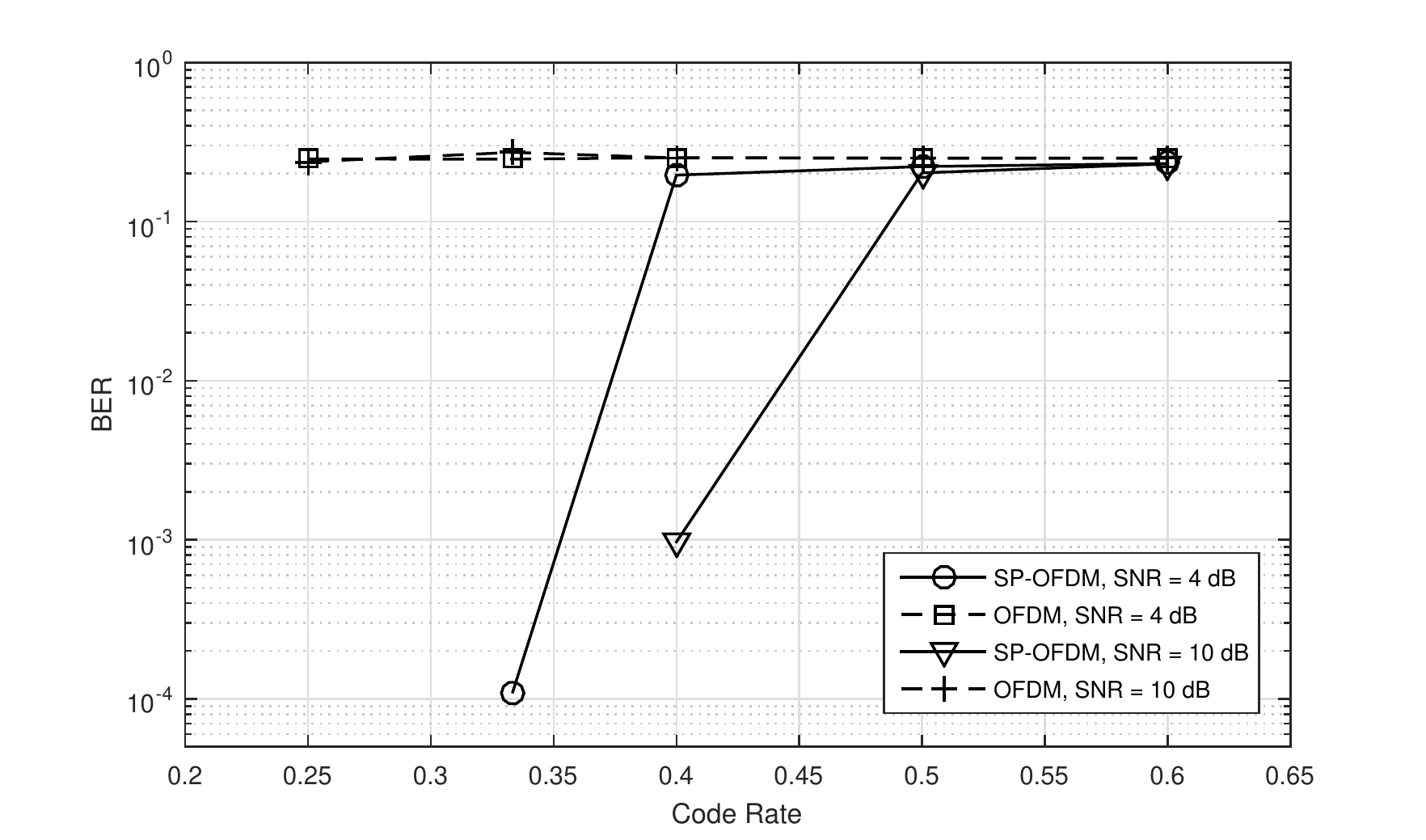}
		\caption{BER performance comparison under disguised jamming in AWGN channels: SP-OFDM versus the traditional OFDM system, signal to jamming power ratio (SJR) = 0 dB.}\label{Fig:BER}	
	\end{figure}
	
	\textbf{Example 3: BER performance under disguised jamming in AWGN channels:} In this example, we analyze the bit error rate (BER) of the proposed system under disguised jamming in AWGN channels. Perfect synchronization is assumed. We use the low density parity check (LDPC) codes for channel coding, and adopt the parity check matrices from the DVB-S.2 standard \cite{Morello2006IEEE}. The coded bits are mapped into QPSK symbols.  The random phase shifts in the proposed secure precoding are approximated as i.i.d. continuous RVs uniformly distributed over $[0, 2\pi)$. We observe that such an approximation has negligible difference on BER performance compared with a sufficiently large $M$. The jammer randomly selects one of the codewords in the LDPC codebook and sends it to the receiver after the mapping and modulation. On the receiver side, we use a soft decoder for the LDPC codes, where the belief propagation (BP) algorithm \cite{Young2001TIT} is employed. The likelihood information in the BP algorithm is calculated using the likelihood function of a general Gaussian channel, where the noise power is set to $1+\sigma^2$ considering the existence of the disguised jamming, and $\sigma^2$ is the noise power. That is, the signal to jamming power ratio (SJR) is set to be 0 dB. It should be noted that for more complicated jamming distributions or mapping schemes, customized likelihood functions basing on the jamming distribution will be needed for the optimal performance.  Fig. \ref{Fig:BER} compares the BERs of the communication system studied with and without the proposed secure precoding under different code rates and SNRs. It can be observed that: (i) under the disguised jamming, in the traditional OFDM system, the BER cannot really be reduced by decreasing the code rate or the noise power, which indicates that without appropriate anti-jamming procedures, the traditional OFDM cannot achieve reliable communications under disguised jamming; (ii) with the proposed SP-OFDM scheme, when the code rates are below certain thresholds, the BER can be significantly reduced with the decrease of code rates using the proposed secure precoding. This demonstrates that the proposed SP-OFDM system can achieve a positive deterministic channel coding capacity under disguised jamming.

	\begin{figure}[t]
		\centering
		\includegraphics[width=\columnwidth]{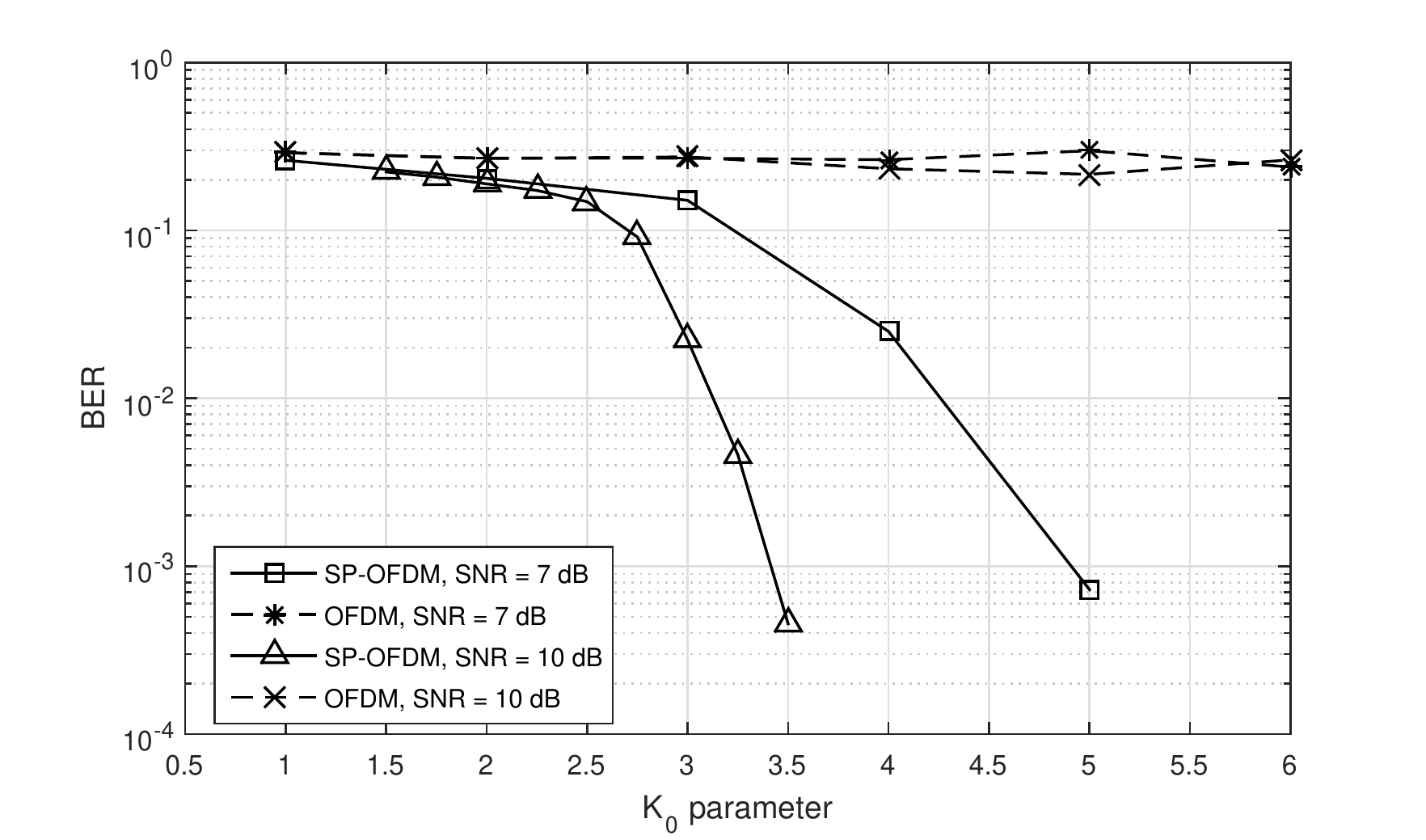}
		\caption{BER performance comparison under disguised jamming in Rician channels: code rate $= 1/3$, SJR = 0 dB.  Here the $K_0$ parameter refers to the power ratio between the direct path and the scattered path.}\label{Fig:Rician}
	\end{figure}

	\textbf{Example 4: BER performance under disguised jamming in Rician channels:} In this example, we verify the effectiveness of the proposed system in fading channels. We consider a Rician channel, where the multipath interference is introduced and a strong line of sight (LOS) signal exists \cite{Rappaport2001}. The fading effect is slow enough so that the channel remains unchanged for one OFDM symbol duration. In the simulation, we set the power of the direct path of Rician channel to be $1$ and vary the $K_0$ parameter, which is the ratio between the power of the direct path and that of the scattered path. Fig. \ref{Fig:Rician} shows the BERs for LDPC code rate $1/3$ under disguised jamming. It can be observed that the proposed system is still effective under the fading channel with a sufficient large $K_0$ parameter. For a small $K_0$ parameter, i.e., when the fading is severe, channel estimation and equalization will be needed to guarantee a reliable communication.    

	\section{Conclusions}\label{Sec:Conclude}
	In this paper, we designed a highly secure and efficient OFDM system under disguised jamming, named securely precoded OFDM (SP-OFDM), by exploiting secure symbol-level precoding basing on phase randomization. We demonstrated the destructive effect of disguised jamming on the traditional OFDM system, and proved the robustness of SP-OFDM against disguised jamming in terms of synchronization and channel capacity.  First, we showed that the traditional OFDM cannot distinguish between the legitimate signal and disguised jamming in the synchronization process, while SP-OFDM, with the secure CP, can achieve accurate synchronization under disguised jamming. Second, we analyzed the channel capacity of the traditional OFDM and the proposed SP-OFDM under hostile jamming using the arbitrarily varying channel (AVC) model. It was shown that the deterministic coding capacity of the traditional OFDM is zero under the worst disguised jamming; on the other hand, with the secure randomness shared between the authorized transmitter and receiver, the AVC channel corresponding to SP-OFDM is not symmetrizable, and hence SP-OFDM can achieve a positive capacity under disguised jamming. Both our theoretical and numerical results demonstrated that SP-OFDM is robust under disguised jamming and frequency selective fading. Potentially, SP-OFDM is a promising modulation scheme for high speed transmission under hostile environments, and the secure precoding scheme proposed in this paper can also be applied to modulation techniques other than OFDM. 
	
	
	\appendix

	\subsection{Proof of Proposition \ref{Prop:pre}}
	\begin{proof}
	Note that $r(t) r^*(t + T_s)$ can be calculated as
	\begin{align}
	&r(t) r^*(t + T_s) = s(t - t_0) s^*(t + T_s - t_0) e^{-j \omega_0 T_s } \nonumber \\
	&~~~~+ z(t) s^*(t + T_s - t_0) e^{-j(\omega_0 t + \omega_0 T_s + \phi_0)} \nonumber \\
	&~~~~+ s(t - t_0) e^{j(\omega_0 t + \phi_0)} z^*(t+T_s) +  z(t)z^*(t+T_s). \label{Eq:rr}
	\end{align}
	In the following we analyze the four terms on the right-hand-side (RHS) of (\ref{Eq:rr}) respectively.
	
	First, define
	\begin{equation}
	Y_{k, 1}(\tau) \overset{\triangle}{=}  \int_{\tau - T_{CP} + k T_b}^{\tau - T_{CP,2} + k T_b} s(t - t_0) s^*(t + T_s - t_0)  \text{d} t, 
	\end{equation}
	for $k \in \mathbb{Z}^*, \tau \in [0, T_b)$. We evaluate the expectation of $Y_{k, 1}(\tau)\tilde{C}^*_{k+d}$ for $d \in \mathcal{K}$. Note that for $t \in [\tau - T_{CP} + k T_b, \tau - T_{CP,2} + k T_b] $, where $\tau \in [0, T_b)$, we have
	\begin{align}
	s(t - t_0) =& \sum_{l = k-1}^{k+1} s_l(t - t_0 - l T_b),  \\
	s(t + T_s - t_0) =& \sum_{l = k-1}^{k+1} s_l(t + T_s - t_0 - l T_b).
	\end{align}
	Note that since the OFDM blocks are zero-mean and independent, for $k_1 \neq k_2$,  we have 
	\begin{equation}
	\mathbb{E} \left\lbrace s_{k_1}(t_1) s^*_{k_2}(t_2) \right\rbrace = 0, \forall t_1, t_2 \in \mathbb{R}.
	\end{equation}
	So we focus on
	{\small\begin{align}
	&\int_{\tau - T_{CP} + k T_b}^{\tau - T_{CP,2} + k T_b} \!\sum_{l = k-1}^{k+1}\! s_l(t - t_0 - l T_b)s^*_l(t + T_s - t_0 - l T_b)~ \text{d}t \nonumber \\
	=&\frac{1}{N_c^2} \sum_{l = -1}^{1} \int_{\tau - l T_b - t_0 - T_{CP} }^{\tau - l T_b -t_0 - T_{CP,2} } \sum_{i_1=0}^{N_c - 1} \tilde{S}_{l+k, i_1} e^{j\frac{2\pi i_1}{T_s} t  } u_{l+k}(t ) \nonumber \\
	&~~~~~~~~~~~~~~~~~~~~~ \sum_{i_2=0}^{N_c - 1} \tilde{S}^*_{l+k, i_2} e^{-j\frac{2\pi i_2}{T_s} t  } u^*_{l+k}(t + T_s ) ~\text{d}t. \label{Eq: Y1_tmp1}
	\end{align}}
	Since for $i_1 \neq i_2$, $\mathbb{E} \{\tilde{S}_{k, i_1} \tilde{S}_{k, i_2}^* \} = 0$, we further focus on 
	\begin{equation}
	\small
	\frac{1}{N_c^2}\sum_{l = -1}^{1} \sum_{i=0}^{N_c - 1} |\tilde{S}_{l+k, i}|^2 \int_{\tau  - l T_b - t_0 - T_{CP} }^{\tau  - l T_b -t_0 - T_{CP,2} }   u_{l+k}(t ) u^*_{l+k}(t + T_s ) ~\text{d}t. \label{Eq: Y1_tmp2}
	\end{equation}
	
	Define function $v_k(\tau)$ as
	\begin{align}
	v_k(\tau)  \overset{\triangle}{=}& \int_{\tau  - T_{CP} }^{\tau  - T_{CP,2} }   u_k(t ) u^*_k(t + T_s ) \text{d}t \nonumber\\
	=& 
	\left\lbrace \begin{array}{rl}
	(\tau + T_{CP, 1}) C_k  , & -T_{CP, 1} \leq \tau < 0, \\
	\tau + (T_{CP, 1} - \tau) C_k, & 0 \leq \tau < T_{CP, 2}, \\
	T_{CP, 2} +  (T_{CP, 1} - \tau) C_k, & T_{CP, 2} \leq \tau < T_{CP, 1}, \\
	T_{CP} - \tau, & T_{CP, 1} \leq \tau < T_{CP}, \\
	0, &  otherwise.
	\end{array}
	\right. 
	\end{align}
	So (\ref{Eq: Y1_tmp2}) can be expressed as
	\begin{equation}
	\frac{1}{N_c^2} \sum_{l = -1}^{1} \sum_{i=0}^{N_c - 1} |\tilde{S}_{l+k, i}|^2 v_{l+k}(\tau - l T_b - t_0).
	\end{equation}
	In addition, since the phase shift symbols are zero-mean and independent, for $\tau  \in \mathbb{R}$,  we have
	\begin{equation}
	\mathbb{E}\{ v_{k_1}(\tau)  C_{k_2}^* \} = \left\lbrace \begin{array}{rl}
	v(\tau),& k_1 = k_2, \\
	0, &  k_1 \neq k_2.
	\end{array} 
	\right. 	
	\end{equation}
	So the expectation of $Y_{k, 1}(\tau) \tilde{C}^*_{k+d}$ is
	\begin{equation}
	\mathbb{E}\{ Y_{k, 1}(\tau) \tilde{C}^*_{k+d} \} =  \left\lbrace 
	\begin{array}{rl}
	\frac{P_S}{N_c}v(\tau + T_b - t_0),& d = k_0 - 1, \\
	\frac{P_S}{N_c}v(\tau - t_0),& d = k_0, \\
	\frac{P_S}{N_c}v(\tau - T_b - t_0),& d = k_0 + 1, \\
	0,& otherwise.				 
	\end{array}
	\right. 
	\end{equation}
	whose maximum is achieved at $\tau = t_0$ and $d = k_0$. In addition, since constellation $\Phi$ is a finite set, the variance of $Y_{k, 1}(\tau) \tilde{C}^*_{k+d}$ is bounded for any possible $k, \tau$ and $d$, while given $\tau$ and $d$, 
	\begin{equation}
	\mathbb{E} \{Y_{k_1, 1}(\tau) \tilde{C}^*_{k_1+d} Y_{k_2, 1}(\tau) \tilde{C}^*_{k_2+d} \} = 0, \text{ for } |k_1 - k_2| > 1.
	\end{equation}
	So as $K \rightarrow \infty$, the variance of $\frac{1}{K} \sum_{k = 0}^{K-1} Y_{k, 1}(t_0) \tilde{C}^*_{k+k_0}$ converges to $0$, and using the Chebychev inequality, we have 
	\begin{equation}
	\frac{1}{K} \sum_{k = 0}^{K-1} Y_{k, 1}(t_0) \tilde{C}^*_{k+k_0} = \frac{P_S T_{CP, 1}}{N_c}, a.s..
	\end{equation} 
	
	Second, define 
	\begin{equation}
	\small
	Y_{k, 2}(\tau) \overset{\triangle}{=} \int_{\tau - T_{CP} + k T_b}^{\tau - T_{CP,2} + k T_b} z(t) s^*(t + T_s - t_0) e^{-j(\omega_0 t + \omega_0 T_s + \phi_0)} \text{d}t,  
	\end{equation}
	and
	\begin{equation}
	\small
	Z_{k, l}(\omega, \tau, t_0) \overset{\triangle}{=} \int_{\tau - T_{CP}}^{\tau - T_{CP,2} } z(t + k T_b)   e^{j\omega t } u_l(t - l T_b + T_s - t_0 )  \text{d}t.
	\end{equation}
	It can be derived that
	\begin{equation}
	\small
	Y_{k, 2}(\tau) =  \sum_{l = -1}^{1}  \sum_{i=0}^{N_c - 1}\frac{  e^{j\frac{2\pi i}{T_s} [  - l T_b - t_0 ]} \tilde{S}_{k+l, i} Z_{k,l}(\frac{2\pi i}{T_s} - \omega_0, \tau, t_0)}{N_c e^{j(k \omega_0 T_b  + \omega_0 T_s + \phi_0)}}. \label{Eq:Y2}
	\end{equation}
	Considering the delay in signal processing, we assume the jamming term $Z_{k,l}(\frac{2\pi i}{T_s} - \omega_0, \tau, t_0)$ is independent of the transmitted symbol $\tilde{S}_{k+l, i}$ in (\ref{Eq:Y2}). Therefore, we have
	\begin{equation}
	\mathbb{E} \{ 	Y_{k, 2}(\tau) \tilde{C}^*_{k+d} \} = 0, \forall k \in \mathbb{Z}^*, \tau \in [0, T_b), d \in \mathcal{K}.
	\end{equation}
	Note that the fourth moment of jamming interference $z(t)$ is bounded, so are the variances of $z(t)$ of $Y_{k, 2}(\tau) \tilde{C}^*_{k+d}$. In addition, 	for $\tau \in [0, T_b), d \in \mathcal{K}$, we have 
	\begin{equation}
	\mathbb{E} \{ 	Y_{k_1, 2}(\tau) \tilde{C}^*_{k_1+d} Y_{k_2, 2}^*(\tau) \tilde{C}_{k_2+d} \} = 0, \forall |k_1 - k_2| > 1.
	\end{equation}
	Therefore,
	\begin{equation}
	\frac{1}{K} \sum_{k = 0}^{K-1} Y_{k, 2}(\tau) \tilde{C}^*_{k+d} = 0, \forall \tau \in [0, T_b), d \in \mathcal{K}, a.s.. \label{Eq:App_tmp1}
	\end{equation}
	
	Third, define 
	\begin{equation}
	Y_{k, 3}(\tau) \overset{\triangle}{=} \int_{\tau - T_{CP} + k T_b}^{\tau - T_{CP,2} + k T_b}  s(t - t_0) e^{j(\omega_0 t + \phi_0)} z^*(t+T_s) \text{d}t.  \nonumber \\
	\end{equation}
	Following the same argument as in the derivation of (\ref{Eq:App_tmp1}) on $Y_{k, 2}(\tau)$, we have
	\begin{equation}
	\frac{1}{K} \sum_{k = 0}^{K-1} Y_{k, 3}(\tau) \tilde{C}^*_{k+d} = 0, \forall \tau \in [0, T_b), d \in \mathcal{K}, a.s..
	\end{equation} 
	
	At last, we define
	\begin{equation}
	Y_{k, 4}(\tau) \overset{\triangle}{=} \int_{\tau - T_{CP} + k T_b}^{\tau - T_{CP,2} + k T_b}  z(t)z^*(t+T_s) \text{d}t.  \nonumber \\
	\end{equation}
	Considering the security of phase shift sequence $C_k$ and the delay in signal processing, we assume that for $t \leq (k+1)T_b + T_s - T_{CP, 2}$, the jammer is unable to recover $\tilde{C}_{k+d}, \forall d \in \mathcal{K}$. Since the fourth moment of $z(t)$ is bounded,  we can have 
	\begin{equation}
	\frac{1}{K} \sum_{k = 0}^{K-1} Y_{k, 4}(\tau) \tilde{C}^*_{k+d} = 0, \forall \tau \in [0, T_b), d \in \mathcal{K}, a.s..
	\end{equation}
	
	In conclusion, by averaging the correlation coefficients $Y_k(\tau, d)$ over multiple OFDM blocks, (\ref{Eq:pre}) can be obtained.
	\end{proof}
	\bibliographystyle{IEEEtran}
	\bibliography{./ref} 

\end{document}